\documentclass[twocolumn]{svjour3}    \usepackage[normal]{optional}

\smartqed  % flush right qed marks, e.g. at end of proof

\usepackage{amsmath}
\usepackage{amssymb}
\usepackage{ifpdf}
\usepackage{cite}
\usepackage{color}
\usepackage{comment}
\usepackage{wrapfig}

\usepackage{url}
\usepackage{amsfonts}

\usepackage[textsize=scriptsize,textwidth=\marginparwidth]{todonotes} 

\usepackage[text={6.5in,9in},centering]{geometry}
\usepackage[colorlinks=true,citecolor=blue,linkcolor=blue]{hyperref}

\usepackage{commands-tam}
\numberwithin{equation}{section}
\numberwithin{theorem}{section}
\spnewtheorem{thm}[theorem]{Theorem}{\bfseries}{\itshape}
\spnewtheorem{clm}{Claim}{\bfseries}{\itshape}
\spnewtheorem{cor}[theorem]{Corollary}{\bfseries}{\itshape}
\spnewtheorem{lem}[theorem]{Lemma}{\bfseries}{\itshape}
\spnewtheorem{prop}[theorem]{Proposition}{\bfseries}{\itshape}
\spnewtheorem{obs}[theorem]{Observation}{\bfseries}{\itshape}
\spnewtheorem{defn}[theorem]{Definition}{\bfseries}{\itshape}
\spnewtheorem{remarkn}[theorem]{Remark}{\bfseries}{\itshape}

\renewcommand{\vec}[1]{\mathbf{#1}}

\newcommand{\polylog}{\mathrm{polylog}}

\usepackage{url}
\usepackage{amsfonts}

\newcommand{\ignore}[1]{}

\usepackage[normalem]{ulem}

\def\longrightharpoonup{\relbar\joinrel\rightharpoonup}
\def\longleftharpoondown{\leftharpoondown\joinrel\relbar}
\def\longrightleftharpoons{\mathop{\vcenter{\hbox{\ooalign{\raise1pt\hbox{$\longrightharpoonup\joinrel$}\crcr\lower1pt\hbox{$\longleftharpoondown\joinrel$}}}}}}
\def\rxn{\mathop{\rightarrow}\limits}  % use as A+B \rxn^k C

\newcommand{\relem}[2]{\paragraph{Lemma #1.}{\emph{#2}}}
\newcommand{\reclaim}[2]{\paragraph{Claim #1.}{\emph{#2}}}
\newcommand{\reobs}[2]{\paragraph{Observation #1.}{\emph{#2}}}

\newcommand{\vx}{\vec{x}}
\newcommand{\ve}{\vec{e}}

\newcommand{\vi}{\vec{i}}

\newcommand{\vp}{\vec{p}}

\renewcommand{\vb}{\vec{b}}
\newcommand{\vc}{\vec{c}}
\newcommand{\vy}{\vec{y}}
\newcommand{\vw}{\vec{w}}
\newcommand{\vz}{\vec{z}}

\newcommand{\vv}{\vec{v}}

\renewcommand{\Pr}{\mathsf{Pr}}
\renewcommand{\time}[2]{\mathsf{T}[#1 \reach #2]}

\usepackage{centernot}
\def\reach{\mathop{\Longrightarrow}\nolimits}

\newcommand{\ldr}{\ell}

\begin{document}

\title{Stable Leader Election in Population Protocols Requires Linear Time
\thanks{A preliminary version of this article appeared as~\cite{LeaderElectionDISC}; the current version has been revised for clarity, and includes several omitted proofs.}}

%\author{
%  David Doty\thanks{California Institute of Technology, Pasadena, CA, USA, {\tt ddoty@caltech.edu}}
%  \and
%  David Soloveichik\thanks{University of California, San Francisco, San Francisco, CA, USA, {\tt david.soloveichik@ucsf.edu}}
%}

    \author{
        David Doty
        \and
        David Soloveichik
    }
    \institute{
        David Doty
        \at
        Department of Computer Science\\
        University of California, Davis\\
        Davis, CA, USA\\
        \email{doty@ucdavis.edu}\\
        Author was supported by NSF grants CCF-1619343, CCF-1219274, and CCF-1162589 and the Molecular Programming Project under NSF grant 1317694.
    \and
        David Soloveichik
        \at
        Department of Electrical and Computer Engineering\\
        University of Texas, Austin\\
        Austin, TX, USA\\
        \email{david.soloveichik@utexas.edu}\\
        Author was supported by an NIGMS Systems Biology Center grant P50 GM081879 and NSF grant CCF-1618895.
    }
    \date{}

% \opt{sub,conf}{
% \author{David Doty\inst{1}\fnmsep\thanks{Author was supported by NSF grants CCF-1219274 and CCF-1162589 and the Molecular Programming Project under NSF grant 1317694.} \and David Soloveichik\inst{2}\fnmsep\thanks{Author was supported by an NIGMS Systems Biology Center grant P50 GM081879 and NSF grant CCF-1442454.}}
% \institute{
% University of California, Davis, Davis, CA, USA, \email{doty@ucdavis.edu}
% \and
% University of Texas at Austin, Austin, TX, USA, \email{david.soloveichik@utexas.edu}
% }
% }

\maketitle

\begin{abstract}
  A population protocol \emph{stably elects a leader} if, for all $n$, starting from an initial configuration with $n$ agents each in an identical state, with probability 1 it reaches a configuration $\vy$ that is \emph{correct} (exactly one agent is in a special leader state $\ldr$) and \emph{stable} (every configuration reachable from $\vy$ also has a single agent in state $\ldr$).
  We show that any population protocol that stably elects a leader requires $\Omega(n)$ expected ``parallel time''---$\Omega(n^2)$ expected total pairwise interactions---to reach such a stable configuration.
Our result also informs the understanding of the time complexity of chemical self-organization
by showing an essential difficulty in generating exact quantities of molecular species quickly.
\end{abstract}

% CoRR abstract
% A population protocol *stably elects a leader* if, for all $n$, starting from an initial configuration with $n$ agents each in an identical state, with probability 1 it reaches a configuration $\mathbf{y}$ that is *correct* (exactly one agent is in a special leader state $\ell$) and *stable* (every configuration reachable from $\mathbf{y}$ also has a single agent in state $\ell$). We show that any population protocol that stably elects a leader requires $\Omega(n)$ expected "parallel time"---$\Omega(n^2)$ expected total pairwise interactions---to reach such a stable configuration. Our result also informs the understanding of the time complexity of chemical self-organization by showing an essential difficulty in generating exact quantities of molecular species quickly.

\newcommand{\ack}{
\paragraph{Acknowledgements.}
The authors thank Anne Condon and Monir Hajiaghayi for several insightful discussions.
We also thank the attendees of the 2014 Workshop on Programming Chemical Reaction Networks at the Banff International Research Station, where the first incursions were made into the solution of the problem of PP stable leader election.
We are also grateful to anonymous reviewers whose comments have significantly improved the presentation.}

\ack

%\opt{sub,conf}{
%  \thispagestyle{empty}
%  \newpage
%  \setcounter{page}{1}
%}

\newcommand{\shrinkBeforePar}{\opt{conf}{\vspace{-5pt}}}

\section{Introduction}\label{sec:intro}

\paragraph{Background.}
Population protocols (PPs) were introduced by Angluin, Aspnes, Diamadi, Fischer, and Peralta\cite{angluin2006passivelymobile} as a model of distributed computing in which the agents have very little computational power and no control over their schedule of interaction with other agents.
They also can be thought of as a special case of Petri nets and vector addition systems\cite{petri1966communication,karp1969parallel}, which were introduced in the 1960s as a model of concurrent processing.
In addition to being an appropriate model for electronic computing scenarios such as mobile sensor networks, they are a useful abstraction of ``fast-mixing'' physical systems such as animal populations\cite{Volterra26},  chemical reaction networks, and gene regulatory networks\cite{bower2004computational}.

A PP is defined by a finite set $\Lambda$ of \emph{states} that each agent may have,\footnote{Some recent work on PPs\cite{TimeSpaceTradeoffsPP, polylogleaderICALP, alistarh2015fast, DBLP:conf/wdag/BeauquierBCS15, AspnesBBS2016, izumi2014space} allows the number of states to grow with the number of agents. 
This paper uses the original model~\cite{angluin2006passivelymobile} with state set that is constant with respect to the population size.
(See section ``Related work''.)} 
together with a \emph{transition function} $\delta: \Lambda \times \Lambda \to \Lambda \times \Lambda$.\footnote{Some work allows ``non-deterministic'' transitions, in which the transition function maps to subsets of $\Lambda \times \Lambda$.
Our results are independent of whether the transition function is deterministic or nondeterministic in this manner.}
Given states $r_1,r_2,p_1,p_2 \in \Lambda$, if $\delta(r_1,r_2)=(p_1,p_2)$ (denoted $r_1,r_2 \to p_1,p_2$) and a pair of agents in respective states $r_1$ and $r_2$ interact, then their states become $p_1$ and $p_2$.\footnote{In the most generic model, there is no restriction on which agents are permitted to interact. If one prefers to think of the agents as existing on nodes of a graph, then it is the complete graph $K_n$ for a population of $n$ agents.}
A \emph{configuration} of a PP is a vector $\vc\in\N^\Lambda$ describing, for each state $s \in \Lambda$, the \emph{count} $\vc(s)$ of how many agents are in state $s$.
Executing a transition $r_1,r_2 \to p_1,p_2$ alters the configuration by decrementing the counts of $r_1$ and $r_2$ by 1 each and incrementing $p_1$ and $p_2$ by 1 each.
Possibly some of $r_1,r_2,p_1,p_2$ are equal to each other, so the count of a state could change by $0$, $1$, or $2$.

Associated with a PP is a set of \emph{valid initial configurations} that we expect the PP to be able to handle.\footnote{The set of valid initial configurations for a ``self-stabilizing'' PP is $\N^\Lambda$, where leader election is provably impossible~\cite{angluin2006self}.
We don't require the PP to work if started in any possible configuration,
but rather allow potentially ``helpful'' initial configurations as long as they don't already have small count states (see ``$\alpha$-dense'' below).}
Agents interact in a pairwise manner and change state based on the transition function.
The next pair of agents to interact is chosen uniformly at random among the $n$ agents.
If no transition rule applies, the interaction is a ``null transition'' $r_1,r_2 \to r_1,r_2$, in which the agents interact but don't change state.
We count the expected number of \emph{interactions} until some event occurs, and then  define the ``parallel time'' until this event as the expected number of interactions divided by the number of agents $n$.
This measure of time is based on the natural parallel model where each agent participates in a constant number of interactions in one unit of time, hence $\Theta(n)$ total interactions are expected per unit time~\cite{angluin2006fast}.
In this paper all references to ``time'' refer to parallel time.

In order to define error-free computation in PPs, we rely on to the model of \emph{stable} computation~\cite{angluin2007computational}.
The model defines computation to be complete when the PP gets to
a configuration that is correct\footnote{What ``correct'' means depends on the task. For computing a predicate, for example, $\Lambda$ is partitioned into ``yes'' and ``no'' voters, and a ``correct'' configuration is one in which every state present has the correct vote.} 
and ``stable'' in the sense that no subsequent sequence of transitions can take the PP to an incorrect configuration.
The model of stable computation disallows error
even in an ``adversarial'' schedule of transitions:
we require that from every configuration reachable by \emph{any} sequence of transitions from the initial configuration, it is possible to reach to a correct stable configuration.
Since the configuration space is finite, stable computation is equivalent to requiring, under the randomized model, that a correct stable configuration is reached with probability $1$.
Thus, although it may appear at first glance that correctness and expected time are defined with respect to different models of transition sequences---adversarial vs random---due to this equivalence they are both seen to be definable in the randomized model.
This notion of stable computation is also equivalent to requiring that every \emph{fair} sequence of transitions reaches a correct stable configuration, where ``fair'' means that every configuration infinitely often reachable is infinitely often reached\cite{angluin2007computational}.
For the arguments of this paper,  the most convenient definition of stable computation is the first one, combinatorial in terms of reachability.

%\todoi{``Stable'' corresponds to ``worst-case correctness'', whereas ``randomized'' corresponds to ``average-case correctness''. Using this language, in this paper we are measuring ``average-case time complexity'' of stable leader election. There is also a notion of ``worst-case time complexity,'' in which strong limitations are known to exist for the class of predicates that can be computed ``quickly'' under this model\cite{SpeedFaultsDISC}. It is straightforward to modify the techniques of\cite{SpeedFaultsDISC} to prove that stable leader election is impossible to do ``quickly'' in that model.}

A PP works ``with a leader'' if there is a special ``leader'' state $\ldr$, and every valid initial configuration $\vi$ satisfies $\vi(\ldr)=1$.
This is in contrast to a uniform initial configuration ($\vi(x)=n$ for some state $x$ and $\vi(y)=0$ for all states $y \neq x$)
or an initial configuration only encoding the input ($\vi(x_i) = n_i$ for $i \in \{1,\ldots,k\}$ to represent any input $(n_1,n_2,\ldots,n_k) \in \N^k$).
It is known that the predicates $\phi:\N^k\to\{0,1\}$ stably computable by PPs are exactly the semilinear predicates, whether an initial leader is allowed or not~\cite{angluin2007computational}.
Although the initial leader does not alter the class of computable predicates,
it may allow faster computation.
For example, the fastest known PPs to stably compute semilinear predicates without a leader take as long as $\Theta(n)$ to converge.\footnote{See ``Open questions'' for the distinction between time to \emph{converge} and time to \emph{stabilize}.
In this paper, the time lower bound we prove is on stabilization.
}
In contrast, with a leader, it is known that any semilinear predicate can be stably computed with expected convergence time $O(\log^5 n)$~\cite{angluin2006fast}.
%\todo{DS: It's a problem that initially we don't distinguish between stabilization and convergence (we say that semilinear computation with a leader is fast, but claim leader election is slow).}
Thus, in certain circumstances, the presence of a initial leader seems to give PPs more computational power (e.g., to converge quickly).
%\footnote{The PP construction in\cite{angluin2006fast} requires expected time $O(\log^5 n)$ to \emph{converge} (stop changing its answer). However, it provably requires expected time $\Omega(n)$ to \emph{stabilize} (reach a configuration from which the answer \emph{cannot} change). In this paper, the time lower bound we prove is on stabilization, not on convergence, which we leave as an open question.}
Angluin, Aspnes, and Eisenstat~\cite{angluin2006fast} asked whether polylogarithmic time stable computation of semilinear predicates is possible without a leader; absent a positive answer, the presence of a leader appears to add power to the model.

\paragraph{Statement of main result.}
Motivated in part by the apparent speedup possible with an initial leader, we ask how quickly a leader may be elected from a configuration lacking one.
%We abstract the problem of leader election away from its apparent uses (such as computing predicates quickly).
We pose the problem as follows: design a PP $\calP$ with two special states $x$ (the initial state) and $\ldr$ (the leader state, which may or may not be identical to $x$) such that, for every $n\in\N$, from the initial configuration $\vi_n$ defined as $\vi_n(x)=n$ and $\vi_n(y)=0$ for all other states $y$, has the following property.
For every configuration $\vc$ reachable from $\vi_n$, there is a configuration $\vy$ reachable from $\vc$ that \emph{has a stable leader},
%\todo{DD: a reviewer was confused by this footnote about ``abstracting''. I think we should add some discussion about the terminating (more common in distributed computing) versus stabilizing convention, and that composition of leader election with something else is trivial in the former but not the latter}
meaning that in all configurations $\vy'$ reachable from $\vy$ (including $\vy$ itself), $\vy'(\ldr) = 1$. %\footnote{Note that this problem abstracts away the idea that the leader might be \emph{useful} for something (such as computing predicates quickly).
% In particular, if a certain PP requires an initial leader, and the correctness of the PP depends on the count of the leader never exceeding 1, prior to the conclusion of the leader election, the presence of multiple leaders may result in unintended transitions.
% However, our main result is an impossibility theorem, showing that even if the objective is simplified to stable leader election, without requiring the leader to be useful for any subsequent task, this \emph{still} requires $\Omega(n)$ time.}

There is a simple $O(n)$ expected time PP for stable leader election, with (assuming $x \equiv \ldr$) the single transition $\ldr,\ldr \to \ldr,f$.
Our main theorem shows that \emph{every} PP that stably elects a leader requires time $\Omega(n)$ to reach a configuration with a stable leader; thus the previous PP is asymptotically optimal.
Section~\ref{sec:difficulty} discusses why some straightforward approaches to proving a time lower bound for leader election fail.

\paragraph{Composing leader election with other tasks.}
We have motivated the importance of leader election in population protocols, in part, by reference to tasks that seem to require a leader to be present initially.
However, a fast leader election protocol would only alleviate the need for an initial leader if the election could be composed with the subsequent task.
Simply combining all transitions for the leader election with the transitions for the task does not necessarily work:
prior to the conclusion of the leader election, the presence of multiple leaders may result in unintended transitions.

A number of ad-hoc schemes successfully compose leader election with computation that requires a leader. 
For example, in ref.~\cite{angluin2006passivelymobile}, 
the stable computation of the ``remainder protocol'' depends on leader election to stabilize on a single leader which collects the remainder information.
As leaders drop out they transfer their information to the surviving leaders.
%Going beyond stable computation to error-prone computation,
%Angluin, Aspnes, and Eisenstat~\cite{angluin2006fast}
%also propose a way to compose leader election with Turing machine simulation that requires a leader.
% (Both their proposed leader election protocol and Turing machine simulation are correct with probability less than $1$.)
% The composition works by detecting when leaders drop out and ``resetting'' the Turing machine simulation.
% When only a single leader is left, no further resetting occurs, and the Turing machine computation proceeds as intended with a single leader.
However,
to facilitate composition of leader election with a subsequent task,
it would be easiest to have a leader election protocol with a stronger termination criterion than the ``stabilizing'' criterion we study, namely the ``terminating'' criterion in which the leader ``knows'' when it has been elected, 
and only then would it trigger the subsequent task to begin.
One way to formalize ``knowing when it has been elected'' is to require that the system \emph{never} has more than one leader.\footnote{If it were possible to detect when a leader election protocol such as $\ldr,\ldr \to \ldr,f$ has stabilized, in the sense that each agent carries a bit $\{g,s\}$ in which all agents start with $g$ and only transition to $s$ after a single $\ldr$ exists, then one could consider the product state $(\ldr,s)$ to be the ``true'' leader state, and $(\ldr,g)$ is considered only a ``candidate'' leader state that may have count $>1$ prior to the stabilization to a single $\ldr$.}
However, it is simple to prove that leader election with this terminating convention is impossible: 
Dividing the initial population in half and preventing the two halves from interacting, 
each half (being a valid initial configuration itself) must elect a separate leader, violating the requirement that its count never exceeds $1$.

Thus the question of how to systematically compose leader election with arbitrary downstream computation is itself open.
Nevertheless, we show that even if the protocol follows the more liberal stabilizing criterion, in which the leader need not ``know'' when it has been elected, stabilization to a single leader \emph{still} requires linear time.

\shrinkBeforePar
\paragraph{Multiple leader states, multiple leaders, and other initial configurations.}
A more general notion of leader election is to identify a subset $\Psi \subset \Lambda$ of states that are all considered leader states, and to require the PP to eventually reach a configuration $\vy$ in which $\sum_{\ldr \in \Psi} \vy(\ldr) = 1$, and this sum is 1 in every configuration reachable from $\vy$.
This corresponds more appropriately to how leader states actually coordinate computation in PPs: a leader agent must remember some state information in between transitions (hence it changes state while remaining the unique leader).
Our techniques actually show this stronger result as well (as explained in Section~\ref{sec:result-statement:locking}). %; we explain in terms of a single leader state for simplicity.
Further, our result implies that a PP cannot elect any fixed quantity of leaders (e.g.~exactly $256$) or variable quantity of leaders under a fixed bound (e.g.~at most $256$) in sublinear expected time.

In the simplest formulation of the task of leader election, we always start with $n$ agents in state $x$ (as described above).
Can we capture more generally leader election from a configuration ``without a pre-existing leader''?
Intuitively, we want to exclude initial configurations with states present in small but non-zero count.
We can exclude such initial configurations,
but allow otherwise deliberately prepared starting conditions,
using the notion of $\alpha$-dense configurations: any state present in the initial configuration has count $\geq \alpha  n$.
Our general negative result (Theorem~\ref{thm:main-q-locking}) implies that even starting with best-case initial configurations,
as long as, for some constant $\alpha > 0$, they are all $\alpha$-dense,
sublinear time leader election is impossible.
An open question relates to weakening the notion of $\alpha$-dense (see below).

\opt{conf}{
\shrinkBeforePar
\paragraph{Why simple proofs fail.}
It is tempting to believe that our negative result could follow by a simple argument based on reasoning about the last transition to change the count of the leader.\footnote{Indeed, if we start with more than one leader, and no transition rule can produce a new leader, then we can easily prove the impossibility of sublinear time leader election as follows.
To quickly reduce from two leaders to one, the other agent's state must be numerous in the population.
Thus, the same transition could occur again, leaving us with no leaders.
}
However, as the following example illustrates, reasoning about the last transition to change the count of the leader is insufficient if some transition can produce a new leader.
Consider the following PP, with initial configuration $\vi$ given by
  $\vi(r) = n^{1/4}$,
  $\vi(x) = n - n^{1/4}$,
and transitions:
\noindent%
\begin{minipage}[t]{0.5\textwidth}%
\vspace{-18pt}%
\begin{align}
  r,r &\to \ldr,k \label{rxn-ldr-produce-1} \\
  r,k &\to k,k \label{rxn-ldr-produce-2}
\end{align}
\end{minipage}\hfill%
\begin{minipage}[t]{0.5\textwidth}%
\vspace{-18pt}%
\begin{align}
  x,k &\to k,k \label{rxn-ldr-produce-3} \\
  \ell,\ell &\to \ell,k \label{rxn-ldr-produce-4}
\end{align}%
\end{minipage}

\vspace{3pt}%
\noindent%
It can be shown (the analysis is presented in the full version of this paper) that this PP stably elects a leader in sublinear time $O(n^{1/2} \log n)$ from the above described non-$\alpha$-dense initial configuration.
Intuitively, it takes expected time $\Theta(n^{1/2})$ for transition~\eqref{rxn-ldr-produce-1} to occur for the first time, producing a single leader.
Transition~\eqref{rxn-ldr-produce-4} ensures that
if transition~\eqref{rxn-ldr-produce-1} occurs more than once,
the PP will eventually stabilize to a single leader.
However, with high probability, transitions~\eqref{rxn-ldr-produce-2} and~\eqref{rxn-ldr-produce-3} consume all $r$ and $x$ in $O(\log{n})$ time \emph{before}~\eqref{rxn-ldr-produce-1} executes a second time.
The probability is high enough that the overall expected time to reach a state with a stable leader is sublinear.
Although the above example does not violate our theorem (since it relies on a non-dense initial configuration), it shows that any proof of the main result cannot be based solely on reasoning about the final transition.
The proof must also effectively establish that configurations, such as the initial configuration of the above PP, cannot be reached with high probability in sublinear time.
}

\shrinkBeforePar
\paragraph{Chemical reaction networks.}
\opt{normal,sub}{
The main result and proof are stated in the language of PPs;
however, the result holds for more general systems that have PPs as a special case.
The discrete, stochastic chemical reaction network (CRN) model has been
extensively used in the natural sciences to model chemical kinetics in a well-mixed solution~\cite{Gillespie77}.
The CRN model is also used prescriptively for specifying the behavior of synthetic chemical systems~\cite{SolSeeWin10,chen2013programmable}.
A CRN is a finite set of \emph{species}  (corresponding to PP states) such as $X,Y,Z$, and \emph{reactions} (corresponding to PP transitions) such as $X+Y \rxn^{k_1} Z$ or $Y \rxn^{k_2} 2X + Z$.
CRNs can be thought of as a generalization of PPs in which spontaneous transitions are possible (unimolecular reactions), the transition may cause the number of agents to change (if the reaction has a different number of products than reactants), and each transition has an associated constant $k$ that affects its probability of being selected.

As an essential form of self-organization,
biological cells seem able to precisely control the count of certain molecules (centriole number~\cite{cunha2009zero} is a well studied example).
How  chemical systems transform  relatively uncontrolled  initial conditions to precisely controlled amounts of desired species is still not well understood.
Our negative result applied to CRNs\footnote{Our result holds for any CRN that obeys Theorem~\ref{thm:timer}, the precise constraints of which are specified in~\cite{DotyTCRN2014} (those constraints automatically apply to all PPs).
Importantly, to generalize to CRNs, we never assume that the count of agents is fixed, but rather use $n$ to indicate the initial count.
}
implies that generating with probability $1$ an exact positive count of a certain species, whether $1$ or $256$, is \emph{necessarily} slower ($\Omega(n)$ time) than, for example, destroying all molecules of the species (through the reaction $X \to \emptyset$),
which takes $O(\log n)$ time.
}

\opt{conf}{
The main result and proof are stated in the language of PPs;
however, the result holds for more general systems that have PPs as a special case.
The discrete, stochastic chemical reaction network (CRN) model has been
extensively used in the natural sciences to model chemical kinetics in a well-mixed solution~\cite{Gillespie77}, and the model is also used prescriptively for specifying the behavior of synthetic chemical systems~\cite{SolSeeWin10,chen2013programmable}.
As an essential form of self-organization,
biological cells seem able to precisely control the count of certain molecules (centriole number~\cite{cunha2009zero} is a well studied example).
How  chemical systems transform  relatively uncontrolled  initial conditions to precisely controlled amounts of desired species is still not well understood.
Our negative result applied to CRNs\footnote{Our result holds for any CRN that obeys Theorem~\ref{thm:timer}, the precise constraints of which are specified in~\cite{DotyTCRN2014} (those constraints automatically apply to all PPs).}
implies that generating with probability $1$ an exact positive count of a certain species, whether $1$ or $256$, is \emph{necessarily} slower ($\Omega(n)$ time) than, for example, destroying all molecules of the species (through the reaction $X \to \emptyset$),
which takes $O(\log(n))$ time.
}

%Our result holds for any CRN that obeys Theorem~\ref{thm:timer} (the precise constraints of which are specified in~\cite{DotyTCRN2014}, but those constraints apply to all PPs).
%\todo{DD: given how we've stated the leader election problem, the valid initial configurations always behave how we want, even if it's a CRN and not a PP}
%\repdd{and such that for infinitely many pairs of valid initial configurations $(\vi,\vi')$, $\vi+\vi'$ is also a valid initial configuration (in our case, this holds for every pair of initial configurations).}{}
%However, the main intellectual challenges of the proof of our main result are fully present in the PP model, so we state the result and proof in that language for the sake of simplicity.

\shrinkBeforePar
\paragraph{Open questions.}
Although we measure computation time with respect to stabilization---the ultimate goal of stable computation---some work uses a different goalpost for completion.
Consider a PP stably electing a leader, and one particular transition sequence that describes its history.
We can say the transition sequence \emph{converged}
at the point when the count of the leader is the same in every subsequently reached configuration (if the PP is correct, this count should be 1).
In contrast, recall that the point of stabilization is when the count of the leader is the same in every subsequently \emph{reachable} configuration (whether actually reached in the transition sequence or not).
%In other words, after stabilization, even an adversarial scheduler cannot change the output.
Measuring time to stabilization in the randomized model, as we do here, measures the expected time until the probability of changing the output becomes $0$.
%We can think of the expected time to stabilize as the time until even an adversary cannot change the output.
To help illustrate the difference between these two subtly different concepts, Section~\ref{sec:difficulty} shows some examples of PPs that converge before stabilizing.

Our proof shows only that stabilization must take expected $\Omega(n)$ time.
However, convergence could occur much earlier in a transition sequence than stabilization.
We leave as an open question whether there is a PP that stably elects a leader and converges in expected $o(n)$ time.
We reiterate that there are PPs that work with a leader to stably compute semilinear predicates with convergence time $O(\log^5 n)$~\cite{angluin2006fast}.
Thus if stable leader election can converge in expected sublinear time,
by coupling the two PPs it might be possible to achieve stable computation of arbitrary semilinear predicates with sublinear convergence time.

It should be noted that the optimal stabilization time for stably computing semilinear predicates, even with an initial leader, is still an open question.
The stably computing PPs converging in $O(\log^5 n)$ time~\cite{angluin2006fast} provably require expected time $\Omega(n)$ to stabilize, and it is unknown whether faster stabilization is possible.
%Allowing an initial leader, the convergence time for each has been shown to be at most $O(\log^5 n)$~\cite{angluin2006fast,CheDotSolDetFuncNaCo}, but these PPs provably require expected time $\Omega(n)$ to stabilize, and it is unknown whether faster stabilization is possible even with an initial leader.

Going beyond stable computation, the open question of Angluin, Aspnes, and Eisenstat~\cite{angluin2006fast} asks whether their efficient high-probability simulation of a space-bounded Turing machine by a PP could remove the assumption of an initial leader.
That simulation has some small probability $\epsilon > 0$ of failure, so if one could elect a leader with a small probability $\epsilon' > 0$ of error and subsequently use it to drive the simulation, 
%\footnote{It is not obvious that one can compose these two tasks. Simply including all the transitions for leader election and for simulating a Turing machine does not work, because the Turing machine simulation assumes there is only one leader, but prior to the conclusion of the leader election, there will be many leaders present, causing errors in the Turing machine simulation. Angluin, Aspnes, and Eisenstat~\cite{angluin2006fast} informally describe a way to detect whenever the number of leaders changes and to restart the simulation from the beginning, cancelling the effect of such errors.}
by the union bound the total probability of error  would be at most $\epsilon + \epsilon'$ (i.e., still close to 0).
However, it remains an open question whether the necessary PP exists.

Our general negative result applies to $\alpha$-dense initial configurations.
However, is sublinear time stable leader election possible from other kinds of initial configurations that satisfy our intuition of not having preexisting leaders?
It is known, for example, that for each $0 < \varepsilon < 1$, an initial configuration with $\Theta(n)$ agents in one state and $\Theta(n^\varepsilon)$ in another state can elect a leader in expected time $O(\log^2 n)$ with high probability\cite{angluin2006fast}, although this protocol has a positive probability of failure.
\opt{normal,sub}{In Section~\ref{sec:difficulty} we give an example PP that stably elects a leader in $O(n^{1/2} \log n)$ time starting from an initial configuration with $\Theta(n)$ agents in one state and $\Theta(n^{1/4})$ in another state.}
\opt{conf}{Above we give an example PP that stably elects a leader (convergence and stabilization) in $O(n^{1/2} \log n)$ time starting from an initial configuration with $\Theta(n)$ agents in one state and $\Theta(n^{1/4})$ in another state.}
In general we want to better characterize the initial configurations for which sublinear time leader election is possible.

\shrinkBeforePar
\paragraph{Related work.}
%\todo{Perhaps add stuff about self-stabilizing leader election.}
Alistarh and Gelashvili\cite{polylogleaderICALP} showed that relaxing the requirement of $O(1)$ states to $O(\log^3 n)$ states allows for a leader to be stably elected
%\todo{Is their PP easily made stable? If so we should replace ``elected with high probability'' with ``stably elected''}
%\repdd{with high probability}{}
in expected time $O(\log^3 n)$.
(Indeed, our proof technique fails if the number of states is not constant with respect to $n$.)
Alistarh, Aspnes, Eisenstat, Gelashvili, and Rivest\cite{TimeSpaceTradeoffsPP} have further refined our understanding of PPs with non-constant states by showing:
(1) a time complexity lower bound: even with up to $O(\log \log n)$ states, any stable leader election protocol still requires ``near linear'' ($\Omega(n / \polylog\ n)$) expected time to stabilize, and
(2) a time complexity upper bound: reducing the $O(\log^3 n)$ state requirement of\cite{polylogleaderICALP} to $O(\log ^2 n)$, at the cost of requiring $O(\log^9 n)$ expected time to stabilize.
They are furthermore able to apply these techniques to show similar lower and upper bounds for the majority problem\cite{AngluinAE2008majority}, where the initial majority among a population of states $x$ and $y$ should eventually occupy the whole population.

Whether a constant bound on the number of states is appropriate depends upon the situation being modeled by the PP.
In some settings---e.g., sensor networks---it is reasonable that employing larger ``swarms'' may be helped by slightly increasing the memory per sensor (say, logarithmically with $n$).
However, when modeling biological regulatory networks, for example, each state corresponds to an existing chemical species, 
and $O(1)$ states is natural.

%Angluin, Aspnes, and Eisenstat~\cite{angluin2006fast} give a leader election protocol that seems intuitively to be correct with high probability, and which works in simulation.

\shrinkBeforePar
\section{Preliminaries}

\shrinkBeforePar
If $\Lambda$ is a finite set (in this paper, of \emph{states}), we write $\N^\Lambda$ to denote the set of functions $\vc:\Lambda \to \N$.
Equivalently, we view an element $\vc\in\N^\Lambda$ as a vector of $|\Lambda|$ nonnegative integers, with each coordinate ``labeled'' by an element of $\Lambda$.
Given $s \in \Lambda$ and $\vc \in \N^\Lambda$, we refer to $\vc(s)$ as the \emph{count of $s$ in $\vc$}.
Let $\|\vc\| = \| \vc \|_1 = \sum_{s\in\Lambda} \vc(s)$ denote the total number of agents.
We write $\vc \leq \vc'$ to denote that $\vc(s) \leq \vc'(s)$ for all $s \in \Lambda$. %, and $\vc < \vc'$ if $\vc \leq \vc'$ and $\vc \neq \vc'$.
Since we view vectors $\vc\in\N^\Lambda$ equivalently as multisets of elements from $\Lambda$, if $\vc \leq \vc'$ we say $\vc$ is a \emph{subset} of $\vc'$.
It is sometimes convenient to use multiset notation to denote vectors, e.g., $\{x,x,y\}$ and $\{2x,y\}$ both denote the vector $\vc$ defined by $\vc(x)=2$, $\vc(y)=1$, and $\vc(z)=0$ for all $z \not\in \{x,y\}$.
Given $\vc,\vc' \in \N^\Lambda$, we define the vector component-wise operations of addition $\vc+\vc'$, subtraction $\vc-\vc'$, and scalar multiplication $m \vc$ for $m \in \N$.
For a set $\Delta \subset \Lambda$, we view a vector $\vc \in \N^\Delta$ equivalently as a vector $\vc \in \N^\Lambda$ by assuming $\vc(s)=0$ for all $s \in \Lambda \setminus \Delta.$
%Write $\vc \rest \Delta$ to denote the vector $\vd \in \N^\Delta$ such that $\vc(s)=\vd(s)$ for all $s \in \Delta$.
%\todo{DD: we only use this in the proof of Claim~\ref{claim:firstfixing}; maybe move it to there.}
%Given $s_1,\ldots,s_k \in \Lambda$, $\vc\in\N^\Lambda$, and $n_1,\ldots,n_k\in\Z$,we write $\vc + \{n_1 s_1,\ldots,n_k s_k\}$ to denote vector addition of $\vc$ with the vector $\vv\in\Z^{\{s_1,\ldots,s_k\}}$ defined by $\vv(s_i)=n_i$.

The following lemma is used frequently in reasoning about population protocols. %, shows that we can always take a infinite nondecreasing subsequence of any infinite sequence.

\begin{lem}[Dickson's Lemma~\cite{dickson}]
\label{lem:dickson}
Any infinite sequence $\vx_0,\vx_1,\ldots \in \N^k$ has an infinite nondecreasing subsequence $\vx_{i_0} \leq \vx_{i_1} \leq \ldots$, where $i_0 < i_1 < ... \in \N$. 
\end{lem}

\opt{normal}{\subsection{Population Protocols}}
%Also, we allow a non-deterministic transition function (i.e., two agents in ``input'' states $a$ and $b$ nondeterministically choose the state pair to which to transition from among a finite set); this difference is inessential since deterministic transition functions can always simulate, through standard techniques, nondeterministic transition functions by using the nondeterminism inherent in the interaction scheduler.

A \emph{population protocol (PP)} is a pair $\calP=(\Lambda,\delta)$,\footnote{We give a slightly different formalism than that of~\cite{angluin2007computational} for population protocols.
The main difference is that since we are not deciding a predicate, there is no notion of inputs being mapped to states or states being mapped to outputs.
Another difference is that we assume the transition function is symmetric (so there is no notion of a ``sender'' and ``receiver'' agent as in~\cite{angluin2007computational}; the unordered pair of states completely determines the next pair of states). 
However,
the results of this paper hold even if we allow the transition function to be non-symmetric or even to be non-deterministic (allowing transitions such as $a,b \to c,d$ and $a,b \to x,y$ to coexist).}
where $\Lambda$ is a finite set of \emph{states}, and $\delta:\Lambda \times \Lambda \to \Lambda \times \Lambda$ is the (symmetric) \emph{transition function}.
A \emph{configuration} of a PP is a vector $\vc \in \N^\Lambda$, with the interpretation that $\vc(s)$ agents are in state $s$.
By convention, the value $n\in\Z^+$ represents the total number of agents $\|\vc\|$.
A \emph{transition} is a 4-tuple $\alpha = (r_1,r_2,p_1,p_2) \in \Lambda^4$,
written $\alpha: r_1,r_2 \to p_1,p_2$,
 such that $\delta(r_1,r_2)=(p_1,p_2)$.
This paper typically defines a PP by a list of transitions, with $\delta$ implicit. 
If an agent in state $r_1$ interacts with an agent in state $r_2$, then they change states to $p_1$ and $p_2$.
For every pair of states $r_1,r_2$ without an explicitly listed transition $r_1,r_2 \to p_1,p_2$, there is an implicit \emph{null} transition $r_1,r_2 \to r_1,r_2$ in which the agents interact but do not change state.

\newcommand{\inputvector}{\vv_{\mathrm{i},\alpha}}
\newcommand{\outputvector}{\vv_{\mathrm{o},\alpha}}
\newcommand{\transitionvector}{\vv_{\alpha}}

More formally, given a configuration $\vc$ and transition $\alpha: r_1,r_2 \to p_1,p_2$, we say that $\alpha$ is \emph{applicable} to $\vc$ if $\vc \geq \{r_1,r_2\}$, i.e., $\vc$ contains 2 agents, one in state $r_1$ and one in state $r_2$.
If $\alpha$ is applicable to $\vc$, then write $\alpha(\vc)$ to denote the configuration $\vc - \{r_1,r_2\} + \{p_1,p_2\}$ (i.e., the configuration that results from applying $\alpha$ to $\vc$); otherwise $\alpha(\vc)$ is undefined.
% More formally, given a configuration $\vc$ and transition $\alpha: r_1,r_2 \to p_1,p_2$, we say that $\alpha$ is \emph{applicable} to $\vc$ if $r_1 \neq r_2$ and $\vc(r_1) \geq 1$ and $\vc(r_2) \geq 1$, or if $r_1=r_2$ and $\vc(r_1) \geq 2$ (i.e., $\vc$ contains 2 agents, one in state $r_1$ and one in state $r_2$).
% Associate with $\alpha$
% the \emph{input vector} $\inputvector \in \Z^\Lambda$ defined by
%   $\inputvector(s) = 1$ if $s = r_1 \neq r_2$ or $s = r_2 \neq r_1$,
%   $\inputvector(s) = 2$ if $s = r_1 = r_2$,
%   and $\inputvector(s) = 0$ otherwise,
% the \emph{output vector} $\outputvector \in \Z^\Lambda$ defined by
%   $\outputvector(s) = 1$ if $s = p_1 \neq p_2$ or $s = p_2 \neq p_1$,
%   $\outputvector(s) = 2$ if $s = p_1 = p_2$,
%   and $\outputvector(s) = 0$ otherwise,
% and the \emph{transition vector} $\transitionvector \in \Z^\Lambda$ defined by
%   $\transitionvector = \outputvector - \inputvector$.
% In other words, $\transitionvector$ represents the amount by which state counts change if transition $\alpha$ is applied.
% If $\alpha$ is applicable to $\vc$, then write $\alpha(\vc)$ to denote the configuration $\vc + \transitionvector$ (i.e., the configuration that results from applying $\alpha$ to $\vc$); otherwise $\alpha(\vc)$ is undefined.
A finite or infinite sequence of transitions $(\alpha_i)$ is a \emph{transition sequence} (or \emph{path}).
Applying a finite or infinite transition sequence $(\alpha_i)$ starting at configuration 
$\vc_0$ induces a finite or infinite sequence of configurations $(\vc_0, \vc_1, \ldots)$ such that, for all $\vc_i$ $(i \geq 1)$, $\vc_i = \alpha_{i-1}(\vc_{i-1})$.\footnote{When the initial configuration to which a transition sequence is applied is clear from context, we may overload terminology and refer to $(\vc_0, \vc_1, \ldots)$ as a transition sequence or path.}
If a finite transition sequence $q$, when applied to the starting configuration $\vc$, ends with $\vc'$, we write $\vc \reach_q \vc'$.
We write $\vc \reach \vc'$ if such a transition sequence exists (i.e., it is possible for the system to reach from $\vc$ to $\vc'$) and we say that $\vc'$ is \emph{reachable} from $\vc$.
If it is understood from context what is the initial configuration $\vi$, then say $\vc$ is simply \emph{reachable} if $\vi \reach \vc$.
Note that this notation omits mention of $\calP$; we always deal with a single PP at a time, so it is clear from context which PP is defining the transitions.
If a transition $\alpha: r_1,r_2 \to p_1,p_2$ has the property that for $i\in\{1,2\}$, $r_i\not\in\{p_1,p_2\}$, or if ($r_1=r_2$ and ($r_i \neq p_1$ or $r_i \neq p_2$)), then we say that $\alpha$ \emph{consumes} $r_i$.
In other words, applying $\alpha$ reduces the count of $r_i$.
We similarly say that $\alpha$ \emph{produces} $p_i$ if it increases the count of $p_i$.

\shrinkBeforePar
\opt{normal}{\subsection{Time Complexity}}

In any configuration the next interaction is chosen by selecting a pair of agents uniformly at random and applying transition function $\delta$.
To measure time we count the expected total number of interactions (including null transitions such as $a,b \to a,b$ in which the agents interact but do not change state),
and divide by the number of agents $n$.
(In the population protocols literature, this is often called ``parallel time''; i.e. $n$ interactions among a population of $n$ agents corresponds to one unit of time).
Let $\vc\in\N^\Lambda$ and $C \subseteq \N^\Lambda$.
Denote the probability that the PP reaches from $\vc$ to some configuration $\vc'\in C$ by $\Pr[\vc \reach C]$.
%We want to talk about the expected time to reach from a configuration $\vc$ to another $\vc'$, but it may be the case that we never reach configuration $\vc'$.
%However, suppose there is a set of configurations $C$, such that from configuration $\vc$, the PP has probability 1 to eventually reach some $\vc' \in C$.
If $\Pr[\vc \reach C]=1$,\footnote{Since PP's have a finite reachable configuration space, this is equivalent to requiring that for all $\vx$ reachable from $\vc$, there is a $\vc' \in C$ reachable from $\vx$.}
define the \emph{expected time to reach from  $\vc$ to $C$}, denoted $\time{\vc}{C}$, to be the expected number of interactions to reach from $\vc$ to some $\vc' \in C$, divided by the number of agents $n$.

\shrinkBeforePar
\section{Main Results}
\label{sec:result-statement}

\shrinkBeforePar
%\subsection{Impossibility of sublinear time stable leader election}
\subsection{Impossibility of Sublinear Time Stable Leader Election}
\label{subsec:stable-leader-election-defn}
We consider the following \emph{stable leader election} problem.
Suppose that each PP $\calP=(\Lambda,\delta)$ we consider has a specially designated state $\ldr\in\Lambda$, which we call the \emph{leader state}.
Informally, the goal of stable leader election is to be guaranteed to reach a configuration with count 1 of $\ldr$ (a leader has been ``elected''), from which no transition sequence can change the count of $\ldr$ (the leader is ``stable'').
We also assume there is a special initial state $x$ (it could be that $x \equiv \ldr$ but it is not required), such that the only valid initial configurations $\vi$ are of the form $\vi(x) > 0$ and $\vi(y) = 0$ for all states $y \in \Lambda \setminus \{x\}$.
We write $\vi_n$ to denote such an initial configuration with $\vi_n(x) = n.$

\begin{defn}\label{defn:stable}
A configuration $\vy$ is \emph{stable} if, for all $\vy'$ such that $\vy \reach \vy'$, $\vy'(\ldr)= \vy(\ldr)$ (in other words, after reaching $\vy$, the count of $\ldr$ cannot change); $\vy$
is said to have a \emph{stable leader} if it is stable and $\vy(\ldr)=1$.
\end{defn}

The following definition captures our notion of stable leader election.
It requires the PP to be ``guaranteed'' eventually to reach a configuration with a stable leader.

\shrinkBeforePar
\begin{defn}\label{defn:stably-elect-leader}
We say a PP \emph{stably elects a leader} if, for all $n\in\Z^+$, for all $\vc$ such that $\vi_n \reach \vc$, there exists $\vy$ with a stable leader such that $\vc \reach \vy$.
\end{defn}

\shrinkBeforePar
In other words, letting $Y$ denote the set of configurations with a stable leader, every reachable configuration can reach to $Y$.
It is well-known~\cite{angluin2007computational} that Definition~\ref{defn:stably-elect-leader} is equivalent to requiring 
$\Pr[\vi_n \reach Y]=1$ for all $n\in\Z^+$.

We note that our PP model captures anonymous nodes defined on a complete communication graph.
Thus, agents in the same state are truly indistinguishable.
Consequently, our formalism does not discern whether the agent that becomes the single leader stays the leader, or whether the leader state moves among agents (reminiscent of token passing).

\shrinkBeforePar
\begin{defn}
Let $t:\Z^+\to\R^+$, and let $Y$ be the set of all configurations with a stable leader.
We say a PP stably elects a leader \emph{in time $t(n)$} if, for all $n \in \Z^+$, $\time{\vi_n}{Y} \leq t(n).$
\end{defn}

%\todo{Add more discussion of this open question in the intro or conclusion.}
%Note that another reasonable definition of the ``time complexity'' of stable leader election would be to consider the expected time until the count of $L$ stops changing.\footnote{Assuming the PP elects a leader stably, this count would always be 1.}
%It is conceivable for a PP to take a different amount of time to stabilize than to stop changing the count of $L$.
\shrinkBeforePar
Our main theorem says that stable leader election requires at least linear time to stabilize:

\shrinkBeforePar
\begin{thm}\label{thm:main}
  If a PP stably elects a leader in time $t(n)$, then $t(n) = \Omega(n).$
\end{thm}

\shrinkBeforePar
Thus a PP that elects a leader in sublinear time cannot do so stably, i.e., it must have a positive probability of failure.

%\todo{The high-level explanation assumes the ``initial'' state is $\vx_m$, not $\vi_n$}
The high-level strategy to prove Theorem~\ref{thm:main} is as follows.
With high probability the PP initially goes from configuration $\vi_n$ to configuration $\vx_n$, such that in the sequence $(\vx_n)$ for increasing population size $n$, the count of each state grows without bound as $n \to \infty$ (indeed the count of each state grows linearly with $n$); this follows from Theorem~\ref{thm:timer}.
We then show that any such configuration must have an ``$O(1)$-bottleneck transition''  before reaching a configuration with a stable leader (informally this means that every transition sequence from $\vx_n$ to a configuration $\vy$ with a stable leader must have a transition in which both input states have count $O(1)$, depending on the PP but not on $n$).
Since it takes expected time $\Omega(n)$ to execute a transition when both states have constant count, from any such configuration it requires linear time to stably elect a leader.
Since one of these configurations is reached from the initial configuration with high probability, those configurations' contribution to the overall expected time dominates, showing that the expected time to stably elect a leader is linear.

\shrinkBeforePar
%\subsection{More general impossibility result in terms of inapplicable transitions and dense configurations}
\subsection{More General Impossibility Result in Terms of Inapplicable Transitions and Dense Configurations}
\label{sec:result-statement:locking}

\newcommand{\bdd}{\mathsf{bdd}}
\newcommand{\unbdd}{\mathsf{unbdd}}

\shrinkBeforePar
Rather than proving Theorem~\ref{thm:main} using the notion of leader stability directly, we prove a more general result concerning the notion of a set of inapplicable transitions.
We generalize in two ways.
(1) A configuration $\vy$ is stable by Definition~\ref{defn:stable} if no transition altering the count of $\ldr$ is applicable in any configuration reachable from $\vy$; Definition~\ref{defn:q-stable} generalizes this to an arbitrary subset $Q$ of transitions.
(2) The valid initial configurations of Section~\ref{subsec:stable-leader-election-defn} are those with $\vi_n(x)=n$ and $\vi_n(y) = 0$ for all $y \in \Lambda \setminus \{x\}$; Theorem~\ref{thm:main-q-locking} generalizes this to any set $I$ of  configurations that are all ``$\alpha$-dense'':
any state can be initially present as long as it is present in ``large count'' (at least a constant fraction of the population; see Definition~\ref{def:alpha-dense}).
We also require the allowed initial configurations to satisfy a weak sort of ``closure under doubling'' property: namely, that 
%\todo{DD:why is this ``for all $\vi$'' and not ``for infinitely many $\vi$''?}
there is an infinite subset $I' \subseteq I$ such that $2I' \subseteq I$. 
In other words, there are infinitely many $\vi \in I$ such that $2 \vi \in I$.
This is true if $I$ is closed under addition, 
for example if $I$ is the set of uniform initial configurations,
or if $I$ is the set of all $\alpha$-dense configurations for some $\alpha > 0$ (since doubling a configuration preserves its density).

%\footnote{\repdd{}{By ``infinitely many $\vi,\vi' \in I$'', we mean that both $\vi$ and $\vi'$ take on an infinite number of values, not merely that the set $\{(\vi,\vi')|\vi,\vi' \in I\}$ is infinite. More formally, the closure property of $I$ we require is that there are two infinite subsets $\{\vi_0,\vi_1,\ldots\} \subseteq I$ and $\{\vi_0',\vi_1',\ldots\} \subseteq I$ such that, for all $m\in\N$, $\vi_m + \vi_m' \in I.$}}
%\repdd{}{More formally, we say $I \subseteq \N^\Lambda$ is \emph{i.o.-closed under addition} if for all $n_0\in\N$ there are $\vi,\vi' \in I$ such that $\|\vi\|,\|\vi'\| \geq n_0$ and $\vi+\vi' \in I.$}
%two infinite subsets $\{\vi_0,\vi_1,\ldots\} \subseteq I$ and $\{\vi_0',\vi_1',\ldots\} \subseteq I$ such that, for all $m\in\N$, $\vi_m + \vi_m' \in I.$}

\shrinkBeforePar
\begin{defn}\label{defn:q-stable}
Let $Q$ be a set of transitions.
A configuration $\vy \in \N^{\Lambda}$ is said to be \emph{$Q$-stable} if no transition in $Q$ is applicable in any configuration reachable from $\vy$.
%For $I \subseteq \N^\Lambda$, let $Y_{I,Q}$ denote the set of $Q$-stable configurations reachable from some configuration in $I$.
\end{defn}

Let $I \subseteq \N^\Lambda$ and $Q$ be a set of transitions.
Let $Y$ be the set of $Q$-stable configurations reachable from some configuration in $I$.
We say that a PP $\calP=(\Lambda,\delta)$ \emph{$Q$-stabilizes from $I$} if, for any $\vi \in I$, $\Pr[\vi \reach Y] = 1$.\footnote{Recall that the condition $\Pr[\vi \reach Y]=1$ is equivalent to $[(\forall \vc \in \N^\Lambda)\ \vi \reach \vc$ implies $(\exists \vy \in Y)\ \vc \reach \vy]$.}
If $I$ and $Q$ are understood from context, we say that $\calP$ \emph{stabilizes}.
For a time bound $t(n)$, we say that $\calP$ stabilizes \emph{in expected time $t(n)$} if, for all $\vi \in I$ such that $\|\vi\|=n$, $\time{\vi}{Y} \leq t(n)$.

\begin{defn} \label{def:alpha-dense}
Let $0 < \alpha \leq 1$.
We say that a configuration $\vc$ is \emph{$\alpha$-dense} if for all $s\in\Lambda$,  $\vc(s) > 0$ implies that $\vc(s) \geq \alpha \|\vc\|$, i.e., all states present in $\vc$ occupy at least an $\alpha$ fraction of the total count of agents.
\end{defn}

In order to reason about the behavior of PPs for larger and larger population sizes, we consider infinite sequences of configurations that ``slice'' across different population counts $n$.
In other words, these sequences consist of configurations satisfying certain criteria that are reachable from ever-larger initial configurations.
(Note that such configurations are not reachable from each other since they have different numbers of agents.)
When $C$ spans infinitely many population sizes, the following definition expresses a basic distinction in how state counts scale with increasing population size:

\begin{defn}
For an (infinite) set/sequence of configurations $C$, let $\bdd(C)$ be the set of states 
$$\setr{s \in \Lambda}{(\exists b\in\N)(\forall \vc \in C)\ \vc(s) < b}.$$
Let $\unbdd(C) = \Lambda \setminus \bdd(C)$.
\end{defn}

\begin{remark}
Note that if $C=(\vc_m)$ is a nondecreasing sequence, then for all $k\in\N$, there is $\vc_m$ such that for all $s \in \unbdd(\vc_m)$, $\vc_m(s) \geq k$. (Note that if $C$ is not nondecreasing, the conclusion can fail; e.g., $\vc_m(s_1)=m, \vc_m(s_2)=0$ for $m$ even and $\vc_m(s_1)=0, \vc_m(s_2)=m$ for $m$ odd.)
\end{remark}

The following is our most general theorem, which the rest of the paper is devoted to proving.

%\begin{thm}\label{thm:main-q-locking}
%  Let $\calP=(\Lambda,\delta)$, let $Q$ be any subset of transitions of $\calP$, let $\alpha > 0$, and let $I\subseteq\N^\Lambda$ be a set of $\alpha$-dense \repds{}{initial} configurations such that, for infinitely many $\vi,\vi' \in I$, $\vi + \vi' \in I$.
%  Suppose $\calP$ $Q$-stabilizes from $I$ in expected time $o(n)$.
%  Then $Y_{I,Q} \cap \N^\Gamma \neq \emptyset$, where $\Gamma = \unbdd(Y_{I,Q})$.
%\end{thm}

\shrinkBeforePar
\begin{thm}\label{thm:main-q-locking}
  Let $\calP=(\Lambda,\delta)$ be a PP, $Q$ be any subset of transitions of $\calP$, $\alpha > 0$, and  $I\subseteq\N^\Lambda$ be a set of $\alpha$-dense initial configurations 
  %\repdd{such that, for infinitely many $\vi,\vi' \in I$, $\vi + \vi' \in I$.}{that is i.o.-closed under addition.}
  such that 
  there is an infinite subset $I' \subseteq I$ such that $2I' \subseteq I$.
  Let $Y$ be the set of all $Q$-stable configurations reachable from $I$. 
  %\todo{DD: we actually prove that a superset $\Delta' \supseteq \Delta$ can be made 0 in $\vy$ (specifically $\Delta' = \bdd((\vy_m))$, where $(\vy_m)$ is one particular sequence of configs in $Y$), but it's okay because $\vy(\Delta')=\vec{0} \implies \vy(\Delta)=\vec{0}$}
  Suppose $\calP$ $Q$-stabilizes from $I$ in expected time $o(n)$.
  Then there are infinitely many $\vv\in Y$ such that $\forall s \in \bdd(Y)$, $\vv(s)=0$.
\end{thm}

\shrinkBeforePar
In other words, if some states have ``small'' count in all reachable stable configurations, then there is a reachable stable configuration in which those states have count 0.
A PP $\calP$ that stably elects a leader is a PP in which $Q$ is the set of transitions that alter the count of $\ldr$, $I = \setr{\vi_n}{n\in\N}$  (note all $\vi_n$ are 1-dense and all of $I$ is closed under doubling), $Y$ is the set of configurations reachable from $I$ with a stable leader, and $\calP$ $Q$-stabilizes from $I$.
Hence by Theorem~\ref{thm:main-q-locking}, if $\calP$ stabilizes in expected time $o(n)$, there is a $\vv$ that is both stable and reachable, where $\vv(\ldr)=0$, a contradiction.
Thus Theorem~\ref{thm:main} follows from Theorem~\ref{thm:main-q-locking}.

We can also use Theorem~\ref{thm:main-q-locking} to prove that stable leader election requires linear time under the more relaxed requirement that there is a set $\Psi \subset \Lambda$ of ``leader states,'' and the goal of the PP is to reach a configuration $\vy$ in which $\sum_{\ldr \in \Psi} \vy(\ldr) = 1$ and stays 1 in any configuration reachable from $\vy$.
Choosing $Q$ as the set of transitions that alter that sum, Theorem~\ref{thm:main-q-locking} implies this form of stable leader election also requires $\Omega(n)$ expected time.

The rest of the paper is organized as follows.
We conclude Section~\ref{sec:result-statement} with two example PPs showing that sublinear time leader election \emph{is possible} if we relax the $\alpha$-dense requirement in Theorem~\ref{thm:main-q-locking}---showing that this requirement is indeed necessary.
These examples are particularly useful in discarding certain simple proof techniques that naturally come to mind. 
Then Section~\ref{sec:tools} develops the technical tools, which we use in Section~\ref{sec:main-proof} to complete the proof of Theorem~\ref{thm:main-q-locking}.
Throughout the rest of this paper, fix $\calP = (\Lambda,\delta)$, $\alpha$, $I$, and $Q$ as in the statement of Theorem~\ref{thm:main-q-locking}.

\subsection{Why Simple Proofs Fail}
\label{sec:difficulty}

It is tempting to believe that the main theorem follows by a simple argument based on reasoning about the last transition to change the count of the leader.
Indeed, if we start with more than one leader, and no transition rule can produce a new leader, then we can easily prove the impossibility of sublinear time leader election as follows.
To quickly reduce from two leaders to one, the other agent's state must be numerous in the population, so the same transition could occur again.
This would leave us with no leaders and no possibility to make a new leader.
However, if transitions \emph{can} produce new leaders, then the argument cannot reason only about the last transition involving the leader.
We illustrate this using two examples, which the authors have found helpful in ruling out plausible-sounding but ultimately insufficient ideas for proving a negative result.
%Indeed, the examples show that we must necessarily limit the configurations that can occur prior to the last transition involving the leader.

%Both of these PPs are good examples of how convergence (changing $\# \ldr$ to 1 for the last time) can strictly precede stabilization (entering a configuration from which it is impossible to change $\# \ldr$).

%
%There are two possibilities for the final transition that alters $\# \ldr$: $r,\ldr \to p_1,p_2$ decreasing $\# \ldr$ from 2 to 1, or $r_1,r_2 \to \ldr,p$ increasing $\# \ldr$ from 0 to 1.

We describe two  PPs that stably elect a leader in sublinear time starting from initial configurations that are not $\alpha$-dense (for $\alpha > 0$ independent of $n$).
(Since the initial configurations are not $\alpha$-dense these PPs do not contradict the statement of our main theorem.)
In both examples, with high probability exactly one transition involving the leader occurs.
In the first example the transition produces precisely one leader in a configuration that previously had none, whereas in the second example, it consumes precisely one leader in a configuration that previously had two.
(Clearly, these are the only two possible forms of the final transition involving the leader.)
The examples imply that any proof of the main result cannot be based solely on reasoning about the final transition, but must additionally establish that configurations such as the initial configurations of these PPs cannot be  reached with high probability in sublinear time.

Consider the following PP, with initial configuration $\vi$ given by
  $\vi(r) = n^{1/4}$,
  $\vi(x) = n - n^{1/4}$,
and transitions:
\begin{eqnarray}
  r,r &\to& \ldr,k \label{rxn-ldr-produce-1}
  \\
  r,k &\to& k,k \label{rxn-ldr-produce-2}
  \\
  x,k &\to& k,k \label{rxn-ldr-produce-3}
  \\
  \ell,\ell &\to& \ell,k \label{rxn-ldr-produce-4}
\end{eqnarray}
Transition~\eqref{rxn-ldr-produce-1} is the only one possible initially, and
it takes expected time $\Theta(n^{1/2})$ to occur for the first time, producing a single leader.
Transition~\eqref{rxn-ldr-produce-4} ensures that
if transition~\eqref{rxn-ldr-produce-1} occurs more than once,
the PP will eventually stabilize to a single leader.
However, with high probability transitions~\eqref{rxn-ldr-produce-2} and~\eqref{rxn-ldr-produce-3} consume all $r$ and $x$ \emph{before}~\eqref{rxn-ldr-produce-1} executes a second time.
After exactly one instance of transition~\eqref{rxn-ldr-produce-1} occurs,
let a \emph{speed fault} denote the event that transition~\eqref{rxn-ldr-produce-1} occurs again (this is the same speed fault concept studied in ref.~\cite{SpeedFaultsDISC}).
%,~\eqref{rxn-ldr-produce-2}, and~\eqref{rxn-ldr-produce-3} are all possible, but transition~\eqref{rxn-ldr-produce-1} occurs first.
For convenience, for state $s \in \Lambda$, let $s$ also denote the count of that state in the configuration considered.
%$\# k = m$, and $\# r = l$

Conditioned on the next interaction being non-null, i.e., it is either a speed fault (transition~\eqref{rxn-ldr-produce-1}) or moves closer to converting all $x$ and $r$ to $k$ (transition~\eqref{rxn-ldr-produce-2} or~\eqref{rxn-ldr-produce-3}),
the probability of a speed fault in any particular configuration is $\frac{r(r-1)}{r(r-1) + 2k(n-k-1)} < \frac{n^{1/2}}{k(n-k-1)}$, since $r(r-1)$ is the number of ways of choosing two agents in state $r$ (leading to a speed fault), and $2k(n-k-1)$ is the number of ways of choosing an agent in state $k$ and another agent in state either $r$ or $x$, when $\ldr=1$ (increasing the count of $k$), and therefore $r + x = n-k-1$.
By the union bound, the probability that a speed fault occurs in between $k = 1$ and $k = n-1$ (at which point transition~\eqref{rxn-ldr-produce-1} is disabled and the PP stabilizes) is at most
\begin{align*}
  n^{1/2} \sum_{k=1}^{n-2} \frac{1}{k (n-k-1)}
  &= n^{1/2} O\left( \frac{\log n}{n} \right)
  \\ &= O\left( \frac{\log n}{n^{1/2}} \right).
\end{align*}

Whether or not a speed fault occurs, to produce $\ldr$, transition~\eqref{rxn-ldr-produce-1} must occur for the first time, taking expected time $O(n^{1/2})$.
%After this, if no speed fault occurs, the expected time to stabilize is at most $O(\log n)$, \repds{}{which is} the expected time for transitions \eqref{rxn-ldr-produce-2} and \eqref{rxn-ldr-produce-3} to consume all $r$ and $x$. %\footnote{Since conditioning on no speed fault can only \emph{reduce} the expected time for transitions~\eqref{rxn-ldr-produce-2} and~\eqref{rxn-ldr-produce-3} to complete, we can use $O(\log n)$ as an upper bound for that conditional expectation.}
Let $T$ be the random variable denoting the time to stabilization \emph{after} transition~\eqref{rxn-ldr-produce-1} has occurred for the first time.
If a speed fault occurs, then transition~\eqref{rxn-ldr-produce-4} must execute enough times to reduce $\ldr$ to 1, which requires expected time $O(n)$~\cite{angluin2006fast}.
Thus $\E[\text{T} | \text{speed fault}]= O(n)$.
Now to analyze $\E[\text{T} | \text{no speed fault}]$,
note that if only transitions~\eqref{rxn-ldr-produce-2} and~\eqref{rxn-ldr-produce-3} existed, then after producing a single $k$, the expected time for transitions~\eqref{rxn-ldr-produce-2} and~\eqref{rxn-ldr-produce-3} to convert all $x$ and $r$ into $k$ would be $O(\log n)$ (this is known as an ``epidemic''~\cite{angluin2006passivelymobile}).
If we consider that transition~\eqref{rxn-ldr-produce-1} is also competing with transitions~\eqref{rxn-ldr-produce-2} and~\eqref{rxn-ldr-produce-3}, and then we condition on transition~\eqref{rxn-ldr-produce-1} not occurring before all $r$ are converted to $k$, then this conditioning can only reduce this expected time.
Thus $\E[T | \text{no speed fault}] = O(\log n).$

Thus, the total expected time to stabilize to a single leader is at most
\begin{eqnarray*}
  && O(n^{1/2}) + \Pr[\text{speed fault}] \cdot \E[\text{T} | \text{speed fault}] 
  \\&&
  + \Pr[\text{no speed fault}] \cdot \E[\text{T} | \text{no speed fault}]
  \\&\leq&
  %O(n^{1/2}) + \Pr[\text{speed fault}] \cdot \E[\text{T} | \text{speed fault}] + 1 \cdot \E[\text{T} | \text{no speed fault}]
  %\\&\leq&
  O(n^{1/2}) + O\left( \frac{\log n}{n^{1/2}} \right) \cdot O(n) + 1 \cdot O(\log n) 
  \\&=&
  O(n^{1/2} \log n),
\end{eqnarray*}
i.e., sublinear time.
%\footnote{Instead of having the sequence of transitions~\eqref{rxn-ldr-produce-2} and~\eqref{rxn-ldr-produce-3} compete with a single execution of transition~\eqref{rxn-ldr-produce-1}, which has weak probability bounds on taking values below its mean, we could have obtained a faster expected time by replacing~\eqref{rxn-ldr-produce-1} with a series of transition executions to which we could apply stronger Chernoff bounds.
%However, our goal is not to optimize these protocols, but instead to show very simple and easy-to-remember protocols that do sublinear time leader election from special configurations, which any impossibility proof needs to rule out.}

The above PP uses a non-dense initial configuration since $\vi(r) = o(n)$.
Thus, although it does not directly contradict the existence of a linear time lower bound from dense configurations, it points out that any proof based on reasoning about the last transition to alter the count of $\ldr$ must disallow the possibility that a leader is elected from some \emph{intermediate} configuration in the manner described above.
With high probability all states obtain count $\Omega(n)$ in a constant amount of time;\footnote{This is the main theorem of~\cite{DotyTCRN2014}, of which Theorem~\ref{thm:timer} is a corollary.} however, it is possible to subsequently reduce some states to sublinear count after super-constant time.
\emph{A priori}, it is conceivable that after, say, $O(\log n)$ time, the PP reaches a non-dense configuration, with $\ldr = 0$ and $r \approx n^{1/4}$ similar to $\vi$ above, which would then elect a leader in sublinear time by producing a single $\ldr$ with high probability.

This example shows the difference between convergence and stabilization.
Assuming no speed fault occurs, the PP converges when the first transition~\eqref{rxn-ldr-produce-1} occurs, but it does not stabilize until transition~\eqref{rxn-ldr-produce-2} reduces the count of $r$ below 2, disabling transition~\eqref{rxn-ldr-produce-1} from occurring again.

We now consider another example PP.
Even if the final change of $\ldr$ takes it from 2 to 1, it is \emph{a priori} conceivable that the final transition $r,\ldr \to p_1,p_2$ to consume $\ldr$ has count $o(n)$ of $r$, so that, although a second execution of the transition is \emph{possible}, the second execution requires sufficiently long  expected time that the system, in the meantime, likely consumes all remaining copies of $r$, along with a mechanism to ensure that a leader is elected even if the second leader is also consumed by an $r$.
The following PP achieves this, with initial configuration $\vi$ given by
  $\vi(\ldr)=2$,
  $\vi(r)=n^{1/2}$,
  $\vi(x)=n-\vi(r)-\vi(\ldr)$,
with transitions
  \begin{eqnarray}
    r,\ldr &\to& r,\ldr' \label{rxn:leader-consume-1}
    \\
    \ldr',x &\to& \ldr',k \label{rxn:leader-consume-2}
    \\
    k,x &\to& k,k \label{rxn:leader-consume-3}
    \\
    k,r &\to& k,k \label{rxn:leader-consume-4}
    \\
    \ldr',\ldr' &\to& \ldr,k \label{rxn:leader-consume-5}
  \end{eqnarray}
An analysis similar to the previous PP shows that the expected time to stabilize to $\ldr=1$ is $O(n^{1/2} \log n)$.
Informally, transition~\eqref{rxn:leader-consume-1} consumes one copy of $\ldr$ after expected time $O(n^{1/2})$.
Transition~\eqref{rxn:leader-consume-2} subsequently produces $k$ in expected time $O(1)$, and transitions~\eqref{rxn:leader-consume-3} and~\eqref{rxn:leader-consume-4} remove all $r$ and $x$ in expected time $O(\log n)$.
With high probability this happens before transition~\eqref{rxn:leader-consume-1} can execute a second time, but if not, then $\ldr' = 2$, so transition~\eqref{rxn:leader-consume-5} guarantees that a single leader is stably elected (as above, if this is needed, it requires expected time $\Omega(n)$,  but it is needed with such low probability that the overall expected time remains sublinear).

In summary, the two PPs above demonstrate that the proof cannot be based solely on reasoning about the final transition to alter $\ldr$, no matter whether that transition increases $\ldr$ from 0 to 1 or decreases it from 2 to 1.

\shrinkBeforePar
%\section{Technical tools}
\section{Technical Tools}
\label{sec:tools}

\shrinkBeforePar
%\subsection{Bottleneck transitions require linear time}
\subsection{Bottleneck Transitions Require Linear Time}

This section proves a straightforward observation used in the proof of our main theorem.
It states that, if to get from a configuration $\vx \in \N^\Lambda$ to some configuration in a set $Y \subseteq \N^\Lambda$, it is necessary to execute a transition $r_1,r_2 \to p_1,p_2$ in which the counts of $r_1$ and $r_2$ are both at most some number $b$, then the expected time to reach from $\vx$ to some configuration in $Y$ is $\Omega(n/b^2)$.

%Let $\vc$ be a configuration and $\alpha: r_1,r_2 \to p_1,p_2$ a transition applicable to $\vc$.

Let $b\in\N$.
We say that transition $\alpha: r_1,r_2 \to p_1,p_2$ is a \emph{$b$-bottleneck} for configuration $\vc$ if $\vc(r_1) \leq b$ and $\vc(r_2) \leq b$.
% Let $\vc,\vc' \in \N^\Lambda$. % be configurations. % such that $\vc \reach \vc'$.
% We say that $\vc$ and $\vc'$ are \emph{$b$-bottleneck separated} if every execution sequence $(\vc=\vc_0,\vc_1,\ldots,\vc_{m-1},\vc_m=\vc')$, with associated transition sequence $(\alpha_0,\alpha_1,\ldots,\alpha_{m-1})$, has the property that there exists $i \in \{0,1,\ldots,m-1\}$ such that $\alpha_i$ is a $b$-bottleneck for $\vc_i$.
% (If $\vc \notreach \vc'$, then they are trivially $b$-bottleneck separated.)
% Let $X,Y \subseteq \N^\Lambda$.
% We say that $X$ and $Y$ are \emph{$b$-bottleneck separated} if, for every $\vx \in X$ and $\vy \in Y$, $\vx$ and $\vy$ are $b$-bottleneck-separated.
% Note that if $X$ and $Y$ are $b$-bottleneck separated, then they are also $b'$-bottleneck separated for all $b' > b$.
% %
% Our $\Omega(n)$ time lower bound relies on the following straightforward observation.
We say a transition sequence $q$ such that $\vx \reach_q \vy$ has a $b$-bottleneck transition if some transition in $q$ is a $b$-bottleneck for the configuration where it is applied.

\newcommand{\obsBottleneck}{
Let $b\in\N$, $\vx \in \N^\Lambda$, and $Y \subseteq \N^\Lambda$ such that $\Pr[\vx \reach Y]=1$.
If every transition sequence taking $\vx$ to a configuration $\vy \in Y$ has a $b$-bottleneck transition, 
%\todo{DS: Still doesn't quite make sense. A transition is a bottleneck only in some configuration. I suggest we only talk about execution sequences and define bottleneck on them. DD: I'm happy with it either way (I think it is obvious from our definition that whether a transition is a bottleneck depends on the config to which it applies, and which config we mean is obvious in this context, but I'm happy if you want to change the defn's also}
then $\time{\vx}{Y} \geq \frac{n-1}{2(b \cdot |\Lambda|)^2}$.}

\shrinkBeforePar
\begin{obs}\label{obs:bottleneck-linear-time}
  \obsBottleneck
\end{obs}

%\todo{As stated this depends on $\Lambda$, but maybe we can eliminate that dependence.}
%\todo{DS: I don't see how to get rid of the dependence on $\Lambda$.}

\newcommand{\proofObsBottleneck}{
\begin{proof}
  The probability that a particular transition $r_1,r_2 \to p_1,p_2$ occurs in any configuration $\vc$ where $\vc(r_1) \leq b$ and $\vc(r_2) \leq b$ is at most $\frac{2b^2}{n(n-1)}$.\footnote{If $r_1 \neq r_2$ and $\vc(r_1)=\vc(r_2)=b$, then the probability to pick the first agent in one of the states $r_1$ or $r_2$ is $\frac{2b}{n}$, and the probability to pick the second agent in the other state is $\frac{b}{n-1}$, so the total probability of both is $\frac{2b^2}{n(n-1)}.$ The case for $r_1=r_2$ gives $\frac{b}{n}$ for the first times $\frac{b-1}{n-1}$ for the second, resulting in lower total probability $\frac{b^2-b}{n(n-1)}$.}
  There are no more than $|\Lambda|^2$ different transitions.\footnote{With a nondeterministic transition function, the total number of transitions would replace the quantity $|\Lambda|^2$ in the conclusion, but it would remain a constant independent of the size of the initial configuration.}
  By the union bound, the probability that in any configuration $\vc$, any $b$-bottleneck transition occurs is no more than $|\Lambda|^2 \frac{2b^2}{n(n-1)}$.
  Thus we can bound the number of interactions until the first $b$-bottleneck transition occurs by a geometric random variable with success probability at most $\frac{2(b \cdot |\Lambda|)^2}{n(n-1)}$, whence the expected number of interactions until the first $b$-bottleneck is at least $\frac{n(n-1)}{2(b \cdot |\Lambda|)^2}$.
  Since the parallel time is defined as the number of interactions divided by $n$, this corresponds to $\frac{n-1}{2(b \cdot |\Lambda|)^2}$ expected time.
  By assumption, the set of configurations $Y$ is reached only after a $b$-bottleneck transition occurs.
  Therefore, the statement of the lemma follows.
\qed\end{proof}
}

\shrinkBeforePar
\opt{normal}{\proofObsBottleneck}
\opt{sub}{A proof of Observation~\ref{obs:bottleneck-linear-time} is given in the appendix.}
\opt{conf}{A proof of Observation~\ref{obs:bottleneck-linear-time} is given in the full version of this paper.}
\opt{sub,conf}{Intuitively, it follows because if two states $r_1$ and $r_2$ have count at most $b$, where $b$ is a constant independent of $n$, then we expect to wait $\Omega(n)$ time before agents in states $r_1$ and $r_2$ interact.}

\shrinkBeforePar
\begin{cor}\label{cor:bottlebeck-linear-time}
  Let $\gamma > 0$, $b\in\N$, $\vc\in\N^\Lambda$, and $X,Y \subseteq \N^\Lambda$ such that $\Pr[\vc \reach X] \geq \gamma$, $\Pr[\vc \reach Y] = 1$, and every transition sequence from every $\vx \in X$ to some $\vy \in Y$ has a $b$-bottleneck transition.
  Then $\time{\vc}{Y} \geq \gamma \frac{n-1}{2(b \cdot |\Lambda|)^2}.$
\end{cor}

\shrinkBeforePar
%\subsection{Sublinear time from dense configurations implies bottleneck free path from configurations with every state ``populous''}
\subsection{Sublinear Time from Dense Configurations Implies Bottleneck Free Path from Configurations with Every State ``Populous''}
The following theorem, along with Corollary~\ref{cor:bottlebeck-linear-time}, fully captures the probability theory necessary to prove our main theorem.\footnote{Theorem~\ref{thm:timer} was proven in a more general model for Chemical Reaction Networks (CRNs)  that obey a certain technical condition~\cite{DotyTCRN2014}. 
As observed in that paper, the class of CRNs obeying that condition includes all PPs, so the theorem holds unconditionally for PPs.
The theorem proved in~\cite{DotyTCRN2014} is more general than Theorem~\ref{thm:timer}, but we have stated a corollary of it here. A similar statement is implicit in the proof sketch of Lemma 5 of a technical report on a variant model called ``urn automata'' that has PPs as a special case~\cite{AngluinUrn03}.
}
Given it and Corollary~\ref{cor:bottlebeck-linear-time}, Theorem~\ref{thm:main-q-locking} is provable (through Lemma~\ref{lem:pos-prob-states-expected-time}) using only combinatorial arguments about reachability between configurations.

For ease of notation, we assume throughout this paper that all states in $\Lambda$ are \emph{producible}, meaning they have positive count in some reachable configuration.
Otherwise the following theorem applies only to states that are actually producible.
%Given initial configuration $\vi \in \N^\Lambda$, let $\Lambda^*_\vi = \setr{x \in \Lambda}{(\exists \vc)\ \vi \reach \vc \text{ and } \vc(x) > 0}$ denote the set of states producible from $\vi$.
%
Recall that for $\alpha > 0$, a configuration $\vc$ is \emph{$\alpha$-dense} if for all $s\in\Lambda$, $\vc(s) > 0$ implies that $\vc(s) \geq \alpha \|\vc\|$.
Say that $\vc \in \N^\Lambda$ is \emph{full} if $(\forall s\in\Lambda)\ \vc(s)>0$, i.e., every state is present.
The following theorem states that with high probability, a PP will reach from an $\alpha$-dense configuration to a configuration in which all states are present (full) in ``large'' count ($\beta$-dense, for some $0 < \beta < \alpha$).

\begin{thm}[adapted from~\cite{DotyTCRN2014}]\label{thm:timer}
  Let $\calP = (\Lambda,\delta)$ be a PP and $\alpha > 0$.
  Then there are constants $\epsilon,\beta > 0$ such that,
letting
$X = \{\ \vx\in\N^\Lambda \ |\ \vx$ is full and $\beta$-dense $\}$,
  for all sufficiently large $\alpha$-dense configurations $\vi$,
  $\Pr[\vi \reach X]$
$\geq 1 - 2^{-\epsilon \|\vi\|}.$
\end{thm}

%In~\cite{DotyTCRN2014}, the theorem is stated for ``sufficiently large'' $\|\vi\|$, but of course one can always choose $\epsilon$ to be small enough to make it true for all $\vi$.

The following lemma reduces the problem of proving Theorem~\ref{thm:main-q-locking} to a combinatorial statement involving only reachability among configurations (and the lack of bottleneck transitions between them).
In Section~\ref{sec:main-proof} we will prove Theorem~\ref{thm:main-q-locking} by showing that the existence of the configurations $\vx_m$ and $\vy_m$ and the transition sequence $p_m$ in the following lemma implies that we can reach a $Q$-stable configuration $\vv \in \N^\Gamma$, where $\Gamma = \unbdd(Y)$ and $Y$ is the set of $Q$-stable configurations reachable from $I$.

% \begin{lem}\label{lem:pos-prob-states-expected-time}
% Let $\alpha > 0$.
%   Let $\calP=(\Lambda,\delta)$ be a PP such that, for some set of transitions $Q$ and infinite set of $\alpha$-dense initial configurations $I$, $\calP$ $Q$-stabilizes from $I$ in expected time $o(n)$.
%   %Let $Y$ be the set of $Q$-stable configurations reachable from $I$.
%   Then for all $m\in\N$, there is a configuration $\vx_m$ reachable from some 
%   $\vi \in I$ and transition sequence $p_m$ such that
%   (1) $\vx_m(s) \geq m$ for all $s \in \Lambda$,
%   (2) $\vx_m \reach_{p_m} \vy_m$, where $\vy_m$ is $Q$-stable, and
%   (3) $p_m$ has no $m$-bottleneck transition.
% \end{lem}

\begin{lem}\label{lem:pos-prob-states-expected-time}
Let $\alpha > 0$.
  Let $\calP=(\Lambda,\delta)$ be a PP such that, for some set of transitions $Q$ and infinite set of $\alpha$-dense initial configurations $I$, $\calP$ $Q$-stabilizes from $I$ in expected time $o(n)$.
  %Let $Y$ be the set of $Q$-stable configurations reachable from $I$.
  Then for all $m\in\N$, 
  there is an $n_0$ such that for all $\vi \in I$ with $\|\vi\| \geq n_0$,
  there is a configuration $\vx_m$ reachable from $\vi$ and transition sequence $p_m$ such that:
  (1) $\vx_m(s) \geq m$ for all $s \in \Lambda$,
  (2) $\vx_m \reach_{p_m} \vy_m$, where $\vy_m$ is $Q$-stable, and
  (3) $p_m$ has no $m$-bottleneck transition.
\end{lem}

% \begin{lem}\label{lem:pos-prob-states-expected-time}
% Let $\alpha > 0$.
%   Let $\calP=(\Lambda,\delta)$ be a PP such that, for some set of transitions $Q$ and infinite set of $\alpha$-dense initial configurations $I$, $\calP$ reaches a set of $Q$-stable configurations $Y$ in expected time $o(n)$.
%   Then there are infinite sequences of configurations $\vx_0,\vx_1,\vx_2,\dots$ and $\vy_0,\vy_1,\vy_2,\dots$ such that
%   for all $\vi \in I$, there is $m \in \N$ and a transition sequence $p_m$ such that
%   (0) $\vi \reach \vx_m$,
%   (1) $\vx_m(s) \geq m$ for all $s \in \Lambda$,
%   (2) $\vx_m \reach_{p_m} \vy_m$, where $\vy_m \in Y$, and % is $Q$-stable, and
%   (3) $p_m$ has no $m$-bottleneck transition.
% \end{lem}

\begin{proof}
  Intuitively, the lemma follows from the fact that states $\vx_m$ are reached with high probability by Theorem~\ref{thm:timer}, and if no paths such as $p_m$ existed, then all paths from $\vx_m$ to a stable configuration would have a bottleneck and require linear time. Since $\vx_m$ is reached with high probability, this would imply the entire expected time is linear.

  For any configuration $\vx_m$ reachable from some configuration in $I$, there is a transition sequence $p_m$ satisfying condition (2) by the fact that $\calP$ $Q$-stabilizes from $I$.
  It remains to show we can find $\vx_m$ and $p_m$ satisfying conditions (1) and (3).

  By Theorem~\ref{thm:timer} there exist $\epsilon,\beta$ (which depend only on $\calP$ and $\alpha$) such that, starting in any sufficiently large initial configuration $\vi$,
  with probability at least $1 - 2^{-\epsilon n}$,
  $\calP$ reaches a configuration $\vx$ where all states $s \in \Lambda$ have count at least $\beta n$, where $n = \|\vi\|$.
  For all $\vi$, let $X_{\vi} = \{\vx\; |\; \vi \reach \vx \text{ and } (\forall s \in \Lambda) \, \vx(s) \geq \beta \|\vi\|\}$.
  Given any $m \in \N$, 
  let $n_0$ be a lower bound on $n$ such that: Theorem~\ref{thm:timer} applies for all $n \geq n_0$,  
$1 - 2^{-\epsilon n_0} \geq \frac{1}{2}$, and further $n_0 \geq m / \beta$.
Then,
%for infinitely many $n \geq n_0$, 
%there is 
for all $\vi \in I$ such that $\|\vi\| = n \geq n_0$,
%and %for all $n \geq n_0$, 
$\Pr[\vi \reach X_{\vi}] \geq \frac{1}{2}$.
  Choose any $n \geq n_0$ for which there is $\vi \in I$ with $\|\vi\|=n$.
  Then any $\vx_m \in X_{\vi}$ satisfies condition (1): $\vx_m(s) \geq m$ for all $s \in \Lambda$.
  We now show that by choosing $\vx_m$ from $X_{\vi}$ for a large enough $n$, we can find a corresponding $p_m$ satisfying condition (3) as well.

  Suppose for the sake of contradiction that for some $m$ we cannot satisfy condition (3) when choosing $\vx_m$ as above, no matter how large we make $n$.
%
%   In other words, there exists an $m\in\N$ such that, for all $n_0\in\N$, 
%   there exists $\vi \in I$ with $\|\vi\|\geq n_0$ such that %$\geq m/\beta$,
%   for all $\vx_m \in X_{\vi}$ reachable from $\vi$ 
%   and for all transition sequences $p_m$ such that 
%   (1) $\vx_m(s) \geq m$ for all $s \in \Lambda$,
%   (2) $\vx_m \reach_{p_m} \vy_m$, where $\vy_m$ is $Q$-stable,
%   we have that (3) does \emph{not} hold, i.e., $p_m$ has an $m$-bottleneck.
%
  This means that, letting $Y$ be set of $Q$-stable configurations, for infinitely many $\vi \in I$, 
  (and therefore infinitely many population sizes $n=\|\vi\|$), 
  all transition sequences from $X_{\vi}$ to $Y$ have an $m$-bottleneck.
  Then Corollary~\ref{cor:bottlebeck-linear-time}, 
  letting $\vc=\vi$, $\gamma=\frac{1}{2}$, $X = X_{\vi}$, and $b=m$, 
  tells us that 
  $\time{\vi}{Y} \geq \frac{1}{2} \frac{n-1}{2(m \cdot |\Lambda|)^2} = \Omega(n)$, a contradiction since $\calP$ is supposed to $Q$-stabilize from $I$ in expected time $o(n)$, i.e., $\time{\vi}{Y} = o(n)$.
\qed\end{proof}

\shrinkBeforePar
%\subsection{Transition ordering lemma}
\subsection{Transition Ordering Lemma}

The following lemma was first proven (in the more general model of Chemical Reaction Networks) in~\cite{SpeedFaultsDISC}.
\opt{normal}{We provide a proof for the sake of self-containment.}%
\opt{sub}{We provide a proof in the appendix for the sake of self-containment.}
Intuitively, the lemma states that a ``fast'' transition sequence (meaning one without a bottleneck transition) that decreases certain states from large counts to small counts must contain transitions of a certain restricted form.
In particular the form is as follows: if $\Delta$ is the set of states whose counts decrease from large to small, then we can write the states in $\Delta$ in some order $d_1,d_2,\ldots,d_k$, such that for each $1 \leq i \leq k$, there is a transition $\alpha_i$ that consumes $d_i$, and every other state involved in $\alpha_i$ is either not in $\Delta$, or comes later in the ordering.
These transitions will later be used to do controlled ``surgery'' on fast transition sequences, because they give a way to alter the count of $d_i$, by inserting or removing the transitions $\alpha_i$, knowing that this will not affect the counts of $d_1,\ldots,d_{i-1}$.

\newcommand{\transitionOrdering}{
  Let $b_1,b_2 \in \N$ such that $b_2 > |\Lambda| \cdot b_1$.
  Let $\vx,\vy \in \N^\Lambda$ such that $\vx \reach \vy$ via transition sequence $q$ that does not contain a $b_2$-bottleneck.
  Define
  \opt{normal}{$$ \Delta = \setr{d \in \Lambda}{\vx(d) \geq b_2 \text{ and } \vy(d) \leq b_1 }. $$}
  \opt{sub,conf}{$\Delta = \{\ d \in \Lambda \ |\ \vx(d) \geq b_2$ and $\vy(d) \leq b_1 \ \}.$}
  Then there is an order on $\Delta$, so that we may write $\Delta = \{d_1,d_2,\ldots,d_k\}$, such that, for all $i \in \{1,\ldots,k\}$, there is a transition $\alpha_i$ of the form
  $d_i,s_i \to o_i,o'_i$, such that $s_i,o_i,o'_i \not \in \{d_{1},\ldots,d_i\}$, and $\alpha_i$ occurs at least $(b_2 - |\Lambda| \cdot b_1)/|\Lambda|^2$ times in $q$.
  % in configurations $\vc$ in which $\vc(S) \geq b_2$. \todo{The consequence that ``...in which $\vc(S) \geq b_2$'' is not used in this paper. }
}

\begin{lem}[Adapted from~\cite{SpeedFaultsDISC}] \label{lem:ordering}
  \transitionOrdering
\end{lem}

\newcommand{\proofTransitionOrdering}{
\begin{proof}
  We define the ordering based on increasing sets $\emptyset = \Delta_0 \subset \Delta_1 \subset \Delta_2 \subset \ldots \Delta_{k-1} \subset \Delta_k = \Delta$, where for each $1 \leq i \leq k$, $\Delta_{i} = \Delta_{i-1} \cup \{d_i\}$.

  We define the ordering inductively ``in reverse,'' by first defining $d_k$, then $d_{k-1}$, etc.
  For all $1 \leq i \leq k$, define $\Phi_i: \N^\Lambda \to \N$ for all configurations $\vc$ by $\Phi_i(\vc) = \sum_{d \in \Delta_i} \vc(d)$.
  $\Phi_k$ is well-defined since $\Delta_k=\Delta$, and $\Phi_i$ is well-defined once we have defined $d_{i+1},\ldots,d_k$, because $\Delta_i = \Delta \setminus \{d_{i+1},\ldots,d_k\}$.

  Because $\vy(d) \leq b_1$ for all $d \in \Delta$, it follows that $\Phi_i(\vy) \leq i \cdot b_1 \leq |\Lambda| \cdot b_1$.
  Recall that $\vx(d) \geq b_2$ for all $d \in \Delta$.
  Let $r$ be the suffix of $q$ after the last configuration $\vc'$ along $q$ such that $\Phi_i(\vc') \geq b_2$.
  Then in all configurations $\vc$ in $r$, $\vc(d) < b_2$ for all $d \in \Delta_i$.
  Because $\Phi_i(\vc') \geq b_2$ but $\Phi_i(\vy) \leq |\Lambda| \cdot b_1$, $r$ contains a subsequence $u$ of transitions, each of which strictly decreases $\Phi_i$, and the total decrease in $\Phi_i$ over all of $u$ is at least $(b_2 - |\Lambda| \cdot b_1)$ between configurations $\vc'$ and $\vy$.

  Let $\alpha: r_1,r_2 \to p_1,p_2$ be a transition in $u$.
  Since $\alpha$ strictly decreases $\Phi_i$, $r_1 \in \Delta_i$ or $r_2 \in \Delta_i$; assume without loss of generality that $r_1 \in \Delta_i$.
  Further, since $u$ does not contain a $b_2$-bottleneck, and all configurations $\vc$ along $u$ have $\vc(d) < b_2$ for all $d \in \Delta_i$, for $\alpha$ not to be a $b_2$-bottleneck, we must have $r_2 \not\in \Delta_i$.
  Since exactly one state in $\Delta_i$ decreases its count, $p_1 \not \in\Delta_i$ and $p_2 \not\in \Delta_i$, or else $\alpha$ would not decrease $\Phi_i$.
  Let $d_i = r_1, s_i = r_2, o_i = p_1$, and $o'_i = p_2$.

  Then $\alpha$ decreases $\Phi_i$ by exactly 1.
  Since there are at least $b_2 - |\Lambda| \cdot b_1$ instances of such transitions in $u$, and there are at most $|\Lambda|^2$ total types of transitions, by the pigeonhole principle at least one transition type must repeat in $u$ at least $(b_2 - |\Lambda| \cdot b_1)/|\Lambda|^2$ times.
\qed\end{proof}
}

\opt{sub,conf}{The intuition behind the proof is that the ordering is given (this is somewhat oversimplified) by the last time in $q$ the state's count drops below $b_2$. Each state in $\Delta$ must go from ``large'' count ($b_2$) to ``small'' count ($b_1$), so when a state $d_i$ is below count $b_2$, if a non-$b_2$-bottleneck transition $d_i,d_j\to\ldots$ occurs, then $d_j$ must exceed $b_2$.
This, in turn, means that state $d_j$ cannot yet have dropped below count $b_2$ for the last time, so $d_j$ is later in the ordering. The full argument is more subtle (and uses a different ordering) because it must establish that the transition's \emph{outputs} in $\Delta$ also come later in the ordering.}

\opt{normal}{\proofTransitionOrdering}

It is instructive to observe how Lemma~\ref{lem:ordering} can fail if the transition sequence $q$ contains a $b_2$-bottleneck.
Consider the linear-time leader election PP given by the transition $\ldr,\ldr \to \ldr,f$ with initial configuration $\vx = \{ n \ldr \}$ and final configuration $\vy = \{1 \ell, (n-1) f \}$.
In this case, $\Delta = \{ \ldr \}$, but once the count of $\ldr$ drops below $b_2$, subsequent transitions are $b_2$-bottlenecks.
Thus, the hypothesis of the lemma is not obeyed, and
indeed, there is no transition $\ldr,s \to o,o'$ such that $\ldr \not\in \{s,o,o'\}$.

\shrinkBeforePar
\section{Proof of Theorem~\ref{thm:main-q-locking}}
\label{sec:main-proof}

%\todoi{DD: one change I made is that rather than using the language of ``pretending to start from configuration $\vx_n$'', I acknowledge that we reach $\vx_m$ from $\vi$, where we have exactly $n$ agents in both states, and $\vx_m(s) \geq m$ for all $s \in \Lambda$, i.e., we have an explicit lower bound on the counts. This means that most of this section doesn't use the variable $n$.}

%\todoi{DD: another change is that rather than using $m$ or $n$ as the bottleneck threshold, I introduce a new variable $b$. Lemma~\ref{lem:pos-prob-states-expected-time} says that for all $b$, for all sufficiently large $m$, we can find $\vx_m$ that will let us prove the theorem. But it means when we pick $b$ here, we need to pick it large enough to satisfy everything down the line (rather than the previous argument, which used $m$ and simply observed that since $m$ grows arbitrarily large, it will be big enough). This is a bit more work but I think it will hopefully be cleaner in the end.}

  Let $I' \subseteq I$ be an infinite subset of $I$ such that $2I' \subseteq I$.
  Recall Lemma~\ref{lem:pos-prob-states-expected-time}.
  We use it (letting $I$ in Lemma~\ref{lem:pos-prob-states-expected-time} be $I'$)
  to construct infinite sequences $(\vx_m)$ and $(\vy_m)$ of configurations and $(p_m)$ of paths as follows.
  Let $m \in \N$.
  Lemma~\ref{lem:pos-prob-states-expected-time} tells us that there is an $n_0$ such that for all $\vi \in I'$ with $\|\vi\| \geq n_0$,
  there is a configuration $\vx_m$ reachable from $\vi$ and transition sequence $p_m$ such that:
  (1) $\vx_m(s) \geq m$ for all $s \in \Lambda$,
  (2) $\vx_m \reach_{p_m} \vy_m$, where $\vy_m$ is $Q$-stable, and
  (3) $p_m$ has no $m$-bottleneck transition.

% \repdd{}{
%   Let $m \in \N$ be odd.
%   Choose $\vi_1,\vi_2 \in I$ with $\|\vi_1\|,\|\vi_2\| \geq n_0$ such that $\vi_1 + \vi_2 \in I$; these exist by the hypothesis of Theorem~\ref{thm:main-q-locking}.
%   Choose $\vx_m$, $p_m$ and $\vy_m$ for $\vi_1$ as above (with $\vi = \vi_1$).
%   Choose $\vx_{m-1}$, $p_{m-1}$, and $\vy_{m-1}$ for $\vi_2$ as above (with $\vi = \vi_2$).
%   This completes the definition of the sequences $(\vx_m)$, $(\vy_m)$, and $(p_m)$.
% }

By Dickson's Lemma (Lemma~\ref{lem:dickson}) there is an infinite subsequence of values of $m$ for which both $(\vx_m)$ and $(\vy_m)$ are nondecreasing.
Without loss of generality, we take $(\vx_m)$, $(\vy_m)$, and $(p_m)$ to be these subsequences.
Let $\Delta = \bdd((\vy_m))$ and 
$\Gamma = \unbdd((\vy_m)) = \Lambda \setminus \Delta$.
Note that since each $\vy_m \in Y$ (the set of stable configurations reachable from some $\vi\in I$), we have that $\bdd(Y) \subseteq \bdd((\vy_m))$.
Thus,
we prove the theorem by showing that for infinitely many $\vy \in Y$, for all $s \in \Delta$, $\vy(s) = 0$.

%To prove Theorem~\ref{thm:main-q-locking} we need to show that there are configurations in $Y$ (the set of $Q$-stable configurations reachable from $I$) that contain states only in $\Gamma$.
Note that stability is closed downward: subsets of a $Q$-stable configuration are $Q$-stable.
For any fixed $\vv^\Gamma \in \N^\Gamma$, $\vv^\Gamma \leq \vy_{m}$ for sufficiently large $m$, by the definition of $\Gamma$ (the states that grow unboundedly in $\vy_m$ as $m\to\infty$). 
All $\vy_m$ are $Q$-stable.
Thus \emph{any} configuration $\vv^\Gamma \in \N^\Gamma$ is automatically $Q$-stable.
This is why Claims~\ref{claim:firstfixing},~\ref{claim:secondfixing}, and~\ref{claim:interleaving} of this proof center around reaching configurations that have count 0 of every state in $\Delta$.

\paragraph{Overview of Claims \ref{claim:firstfixing}--\ref{claim:interleaving}.}
Recall the path $\vx_m \reach_{p_m}  \vy_m$ from Lemma~\ref{lem:pos-prob-states-expected-time}.
Intuitively, Claim~\ref{claim:firstfixing} below says that because this path is $m$-bottleneck free,
Lemma~\ref{lem:ordering} applies, and its transitions can appended to the path to consume all states in $\Delta$ from $\vy_m$, resulting in a configuration $\vz_m^\Gamma$ that contains only states in $\Gamma$.
If this is possible directly as stated, this would correspond to the formal claim that $\vx_m \reach_{p_m} \vy_m \reach_{p'_m} \vz_m$, where $\vz_m \in \N^\Gamma$ contains no states in $\Delta$.
However, we do not know how to prove this directly (although it may be true).
Instead, we show in Claim~\ref{claim:firstfixing} that $\vy_m$ can reach to such a  $\vz_m^\Gamma \in \N^\Gamma$ \emph{if some extra agents in special states are supplied}, i.e., that there exists $\ve \in \N^\Lambda$ (the extra agents) such that $\vy_m + \ve \reach \vz_m^\Gamma$.
By additivity $\vx_m + \ve \reach_{p_m} \vy_m + \ve$, so $\vx_m + \ve \reach \vz_m^\Gamma$.

So where will these extra agents come from?
Although we talk about them as if they are somehow physically added, in actuality, we'll start with a larger initial configuration and ``guide'' some of the agents to the desired states that make up $\ve$.
Claims~\ref{claim:secondfixing} and \ref{claim:interleaving} explain how this happens.

Claim~\ref{claim:secondfixing}, also relying on Lemma~\ref{lem:ordering}, 
is a way to produce the states corresponding to $\ve$ needed for Claim~\ref{claim:firstfixing}.
Claim~\ref{claim:secondfixing} states, intuitively, that any such $\ve \in \N^\Lambda$ can be produced from $\vx_m$, as long as this is done ``in the context'' of extra states (corresponding to $\vp\in\N^\Lambda$).
However, unlike the extra states $\ve\in\N^\Lambda$ as used in Claim~\ref{claim:firstfixing}, the states in $\vp$ are ``recovered''.

To understand why Claim~\ref{claim:secondfixing} is useful, it is helpful to look how the proof of Claim~\ref{claim:interleaving} works.
The initial configuration $\vi$ in Claim~\ref{claim:interleaving} can be split into two halves $\vi'$ where $\vi = 2\vi'$.
We have $\vi' \reach \vx_m$, so $\vi \reach 2 \vx_m$.
The proof of Claim~\ref{claim:interleaving} works by employing Claim~\ref{claim:secondfixing} to produce $\ve$ from one copy of $\vx_m$ (while at the same time turning $\vx_m$ into $\vw^\Gamma \in \N^\Gamma$), using the other copy of $\vx_m$ as the ``context'' $\vp$ that allows Claim~\ref{claim:secondfixing} to work.
Once $\ve$ is produced and the second copy of $\vx_m$ is recovered, Claim~\ref{claim:firstfixing} is then used to show that $\vx_m +\ve \reach \vz^\Gamma$. 
In other words, having already eliminated $\Gamma$ states from the first copy of $\vx_m$, turning it into $\vw^\Gamma$, we use $\ve$ to eliminate $\Gamma$ states from the other copy of $\vx_m$, turning it into $\vz^\Gamma$.
Thus 
$\vi = 2\vi' \reach 2 \vx_{m} \reach \vx_{m} + \vw^\Gamma + \ve \reach \vz^\Gamma + \vw^\Gamma = \vv^\Gamma \in \N^\Gamma$.
We argued above that $\vv^\Gamma$ is $Q$-stable, proving Theorem~\ref{thm:main-q-locking}.

\begin{clm}   \label{claim:firstfixing}
There is $\ve \in \N^\Lambda$ such that for all large enough $m$,
there is $\vz^\Gamma_m \in \N^\Gamma$,
such that $\vx_m + \ve \reach \vz^\Gamma_m$.
\end{clm}

\paragraph{Example.}
\newcommand{\hd}{\hat{d}}
We first illustrate Claim~\ref{claim:firstfixing} through an example.
Define a PP by the transitions
%\begin{eqnarray}
%  b,a &\to& f,c  \label{rxn:claims-example-1}
%  \\
%  f,c &\to& f,a  \label{rxn:claims-example-2}
%  \\
%  a,c &\to& f,f  \label{rxn:claims-example-3}
%  \\
%  b,c &\to& f,b  \label{rxn:claims-example-4}
%  \\
%  f,b &\to& f,f  \label{rxn:claims-example-5}
%\end{eqnarray}
\opt{normal}{
\begin{eqnarray}
  b,a &\to& f,c  \label{rxn:claims-example-1}
  \\
  b,c &\to& f,a  \label{rxn:claims-example-2}
  \\
  a,c &\to& f,f  \label{rxn:claims-example-3}
  \\
  f,c &\to& f,b  \label{rxn:claims-example-4}
  \\
  f,b &\to& f,f  \label{rxn:claims-example-5}
\end{eqnarray}
}
\opt{sub,conf}{

\noindent%
\begin{minipage}[t]{0.5\textwidth}%
\vspace{-18pt}%
\begin{align}
  b,a &\to f,c  \label{rxn:claims-example-1}
  \\
  b,c &\to f,a  \label{rxn:claims-example-2}
  \\
  a,c &\to f,f  \label{rxn:claims-example-3}
\end{align}%
\end{minipage}\hfill%
\begin{minipage}[t]{0.5\textwidth}
\vspace{-18pt}%
\begin{align}
  f,c &\to f,b  \label{rxn:claims-example-4}
  \\
  f,b &\to f,f  \label{rxn:claims-example-5}
\end{align}
\end{minipage}
\vspace{2pt}

\noindent}For convenience, for state $s \in \Lambda$, let $s$ also denote the count of that state in the configuration considered.
Let configuration $\vx_m$ be where $f=100$, $a=100$, $b=100$, $c=100$.
Suppose a transition sequence $p_m$ without an $m$-bottleneck ($m=100$) takes the PP from $\vx_m$ to $\vy_m$, in which $a = 3$, $b=2$, $c=1$, and $f=394$.
Then in the language of Lemma~\ref{lem:ordering}, $\Delta=\{a,b,c\}$; these states go from ``large'' count in $\vx_m$ to ``small'' count in $\vy_m$.

Our strategy is to add interactions to $p_m$ in order to reach a configuration $\vz_m^\Gamma$ with $a=b=c=0$.
There are two issues we must deal with.
First, to get rid of $a$ we may try to add $3$ instances of \eqref{rxn:claims-example-1} at the end of $p_m$.
However, there is only enough $b$ for $2$ instances.
To eliminate such dependency,
in Claim~\ref{claim:firstfixing},
whenever we add a transition %$r_1,r_2 \to p_1,p_2$,
$b,a \to f,c$,
we add an extra agent in state $b$ to $\ve$.
(In general if we consume $r_2$ by adding transition $r_1,r_2 \to p_1,p_2$, we add an extra agent in state $r_1$ to $\ve$.)
%For brevity of notation, for a state $s \in \Lambda$, we use $s$ to refer also to the count of $s$ in the configuration understood from context.
Second, we need to prevent circularity in consuming and producing states.
Imagine trying to add more executions of~\eqref{rxn:claims-example-1} to get $a$ to 0 and more of~\eqref{rxn:claims-example-2} to get $c$ to 0; this will fail because these transitions conserve the quantity $a+c$.
To drive each of these states to 0, we must find some ordering on them so that each can be driven to 0 using a transition that does not affect the count of any state previously driven to 0.

Lemma~\ref{lem:ordering} gives us a way to eliminate such dependency systematically.
In the example above,
we can find the ordering $d_1 \equiv a$, $d_2 \equiv c$, and $d_3 \equiv b$,
with respective transitions~\eqref{rxn:claims-example-1} to drive $a$ to 0 (3  executions), \eqref{rxn:claims-example-4} to drive $c$ to 0 (4  executions: 1 to consume the 1 copy of $c$ in $\vy_m$, and 3 more to consume the extra 3 copies that were produced by the 3 extra executions of~\eqref{rxn:claims-example-1}), and~\eqref{rxn:claims-example-5} to drive $b$ to 0 (6 executions: 2 to consume 2 copies of $b$ in $\vy_m$, and 4 more to consume the extra 4 copies that were produced by the 4 extra executions of~\eqref{rxn:claims-example-4}).

%\vspace{10pt}

\begin{proof}[of Claim~\ref{claim:firstfixing}]
Intuitively, the proof works as follows.
Recall that $\vx_m \reach_{p_m}  \vy_m$ and $p_m$ does not contain an $m$-bottleneck.
The goal is to get from configuration $\vy_m$ (which may be positive on some elements of $\Delta$) to $\vz^\Gamma_m$ (which is 0 on all elements of $\Delta$).
(Recall that $\Delta$ and $\Gamma$ partition the set of states $\Lambda$.)
We will show that we can append to the end of $p_m$ transitions $\alpha_i: d_i,s_i \to o_i,o'_i$, for $i \in \{1,\ldots,k\}$---in that order---such that for all $i$, $d_i \in \Delta$ and $o_i,o'_i \not\in \{d_1,\ldots,d_i\}$. %\todo{For this lemma we don't need that $s_i\not\in \{d_1,\ldots,d_i\}$; it might be clearer not to include it in the list.}
We use Lemma~\ref{lem:ordering} to find the necessary transitions.
We add enough $\alpha_i$ transitions to consume all copies of $d_i$.
(We'll use $c_i$ to indicate the number of transitions added.)
However, this will also consume copies of $s_i$, so we add more agents in state $s_i$ to $\ve$ to account for this.
(We'll use $\ve_i$ to represent the additional copies of $s_i$ added, with eventually $\ve = \sum_{i=1}^k \ve_i$.)
Once we have added enough $\alpha_i$ transitions to make the count of $d_i$ equal to 0, by the fact that for all $j$, $o_j,o'_j \not\in \{d_1,\ldots,d_j\}$, subsequently added transitions $\alpha_{j}$ for $j > i$ will not produce $d_i$ (so its count will stay 0), nor will it require consuming $d_i$ (so the transitions will be applicable).
However, prior to reaching the point where we add $\alpha_i$ transitions, if $d_i = o_j$ or $o'_j$ for $j<i$, then the excess copies of $d_i$ generated by the extra $\alpha_j$ transitions mean that we may need to add more than $\vy_m(d_i)$ copies of $\alpha_i$ to consume all the copies of $d_i$.
The resulting configuration will be $\vz^\Gamma_m \in \N^\Gamma$.

More formally, we
choose large enough $m$ such that the counts of species in $\Delta$ are no longer changing with $m$ in $(\vy_m)$;
thus, the same $\ve$ and the same appended transitions will suffice for all larger $m$.
Recall that $\vx_m \reach_{p_m}  \vy_m$ and $p_m$ does not contain an $m$-bottleneck.
We'll apply Lemma~\ref{lem:ordering} on this path with $b_2 = m$ and let $b_1$ be the largest count of any species in $\Delta$ anywhere in the sequence $(\vy_m)$.
(If necessary, increase $m$ further to ensure $b_2 > |\Lambda| \cdot b_1$.)
Note that with these parameters, $\Delta = \bdd((\vy_m))$ exactly matches the set of states ``$\Delta$'' defined in the statement of Lemma~\ref{lem:ordering}.
Then this lemma tells us that there is an ordering on $\Delta$, so that we can write $\Delta=\{d_1,\ldots,d_k\}$, such that for each $1 \leq i \leq k$, there is a transition $\alpha_i: d_i,s_i \to o_i,o'_i$ such that $d_i \in \Delta$ and $s_i,o_i,o'_i \not\in \{d_1,\ldots,d_i\}$. 
(The fact that $s_i \not\in \{d_1,\ldots,d_i\}$ is not used in this Claim, but it will be essential for Claim~\ref{claim:secondfixing}).

% We determine $\ve$ and show how to consume all states in $\Delta$ using the above transitions in the following iterative procedure.
% %Start with $\ve_0 = \vec{0}$ and $\vo_0=\vec{0}$.
% Intuitively, $\ve_i$ represents the total number of transitions, among $\alpha_1,\ldots,\alpha_{i-1}$, that we add that have $d_i$ as an input state.
% Therefore, that many extra copies of $d_i$ must be present in $\ve$ to allow us to execute all of those transitions starting at the configuration $\vy_m + \ve$.
% $\vo_i$ represents the total number of transitions, among $\alpha_1,\ldots,\alpha_{i-1}$, that we add that have $d_i$ as an \emph{output} state (counting twice each time a transition has $d_i$ as both output states).
% Therefore, when we apply reaction $\alpha_i$ in order to consume all of $d_i$, we need to consume not just $\vy_m(d_i)$, but $\vy_m(d_i) + \vo_{i-1}(d_i)$, since the previously added transitions that produce $d_i$ increased its count.
% Thus we add $c_i = \vy_m(d_i) + \vo_{i-1}(d_i)$ executions of $\alpha_i$.

%Given $s \in \Lambda$ and $c\in\N$, we write $\{c s\}$ to denote the vector $\vv \in \N^\Lambda$ defined by $\vv(s)=c$ and $\vv(s') = 0$ for all $s' \in \Lambda \setminus \{s\}$.
For all $i\in\{1,\ldots,k\}$, let
\begin{itemize}
 \item $c_i = \vy_m(d_i) + \sum_{j=1}^{i-1} \{c_j o_j,c_j o'_j\}$
 (note that if $o_j=o'_j$, then $\{c_j o_j,c_j o'_j\} = \{2c_j o_j\}$)
 \item $\ve_i = \{c_i s_i\}$
\end{itemize}

Given a transition $\alpha$ and $j\in\N$, let $j \cdot \alpha$ denote the transition sequence consisting of $j$ copies of $\alpha$.
%
% Let $\ve = \sum_{i=1}^k \ve_i$.
% For all $i\in\{0,\ldots,k\}$, define $p_{m,i}$ inductively to be $p_{m,i-1} \oplus (c_i \cdot \alpha_i)$, where the base case is $p_{m,0} = p_m$, and define $\vz_{m,i}$ inductively to be such that $\vx_m + \sum_{j=1}^i \ve_i \reach_{\p_{m,i}} \vz_{m,i}$, where the base case is $\vz_{m,0} = \vy_m$.
% Then $\vz^\Gamma_m$ in the statement of the claim will be $\vz_{m,k}$.
%
For all $i\in\{0,\ldots,k\}$, define $\vz_{m,i}$ as follows,
where $\vz_{m,0} = \vy_m$:
\begin{align*}
  %\vx_m &\reach_{\p_m} \vy_m = \vz_{m,0} \\
  \vz_{m,0} + \ve_1 &\reach_{c_1 \cdot \alpha_1} \vz_{m,1} \\
  \vz_{m,1} + \ve_2 &\reach_{c_2 \cdot \alpha_2} \vz_{m,2} \\
  \ldots \\
  \vz_{m,k-1} + \ve_k &\reach_{c_k \cdot \alpha_k} \vz_{m,k}.
   %(= \vz_m^\Gamma)
\end{align*}

%Given transition sequences $p$ and $q$, let $p \oplus q$ denote $q$ appended to the end of $p$.
For all $i\in\{0,\ldots,k\}$, define path $p_{m,i}$ inductively to be $p_{m,i-1}$ followed by  $c_i \cdot \alpha_i$, where the base case is $p_{m,0} = p_m$, so that $\vx_m + \sum_{j=1}^i \ve_j \reach_{p_{m,i}} \vz_{m,i}$.
We prove by induction on $i$ that:
\begin{enumerate}
  \item\label{ind:pmvalid} $p_{m,i}$ is a valid transition sequence (i.e., it never has a transition in a configuration in which the input states are not present),
 \item\label{ind:zd0} for all $j \in \{1,\ldots,i\}$, $\vz_{m,i}(d_j)=0$, and
 \item\label{ind:zdc} for all $j \in \{i+1,\ldots,k\}$, $\vz_{m,i}(d_j) = \vy_m(d_j) + \left(\sum_{\ell=1}^{i} \{c_\ell o_\ell,c_\ell o'_\ell\} \right)(d_j)$ (in particular, $\vz_{m,i}(d_{i+1}) = c_{i+1}$, which is the amount that we remove in path $c_{i+1} \cdot \alpha_{i+1}$).
\end{enumerate}

The base case $i=0$ for~\ref{ind:pmvalid} follows from the fact that $p_m$ is a valid transition sequence to apply to $\vx_m$.
The base case is vacuous for~\ref{ind:zd0}.
For~\ref{ind:zdc}, observe that $\vz_{m,0}=\vy_m$ and the sum is empty when $i=0$.

Assume the inductive case for $i-1$.
% Then by induction hypothesis~\ref{ind:zd0}, $\vz_{m,i-1}(d_j)=0$ for all $j \in \{1,\ldots,i-1\}$.
% Since $\alpha_i$ does not have $d_j$ as an input or output state, and because $\ve_i = \{c_i s_i\}$ is 0 on $d_j$ (since $s_i \not\in \{d_1,\ldots,d_{i-1}\}$), it follows that $\vz_{m,i}(d_j)=0$ for all $j \in \{1,\ldots,i-1\}$.
% By induction hypothesis~\ref{ind:zdc}, $\vz_{m,i-1}(d_i) = c_i$, and since $c_i \cdot \alpha_i$ consumes exactly $c_i$ copies of $d_i$, this shows that $\vz_{m,i}(d_i) = 0$, establishing the inductive case for~\ref{ind:zd0}.
%
% ``$\alpha_i$ does not have $d_j$ as an input or output state'' = ``$d_j \not \in \{s_i,o_i,o_i',d_i\}$''
%
We have $\vz_{m,i} = \vz_{m,i-1} + \ve_i + c_i \{o_i,o_i'\} - c_i \{d_i,s_i\}$
where $c_i \{o_i,o_i'\} - c_i \{d_i,s_i\}$ is the effect of applying transition $\alpha_i$ for $c_i$ times.
Since $\ve_i = \{c_i s_i\}$, we have
$\vz_{m,i} = \vz_{m,i-1} + c_i \{o_i,o_i'\} - c_i \{d_i\}$.
By the induction hypothesis~\ref{ind:zdc} for $i-1$,
$\vz_{m,i-1}(d_i) = c_i$.
Since $d_i \not \in \{o_i,o_i'\}$,
this implies that $c_i \alpha_i$ is a valid path (which establishes the inductive case for~\ref{ind:pmvalid}).
Further, for all $j \in \{1,\dots,i-1\}$, 
$d_j \not \in \{o_i,o_i',d_i\}$,
which implies that 
the amount of $d_j$ is not changed by $\alpha_i$, 
and by inductive hypothesis~\ref{ind:zd0}, for all $j \in \{1,\dots,i-1\}$, $\vz_{m,i}(d_j) = 0$.
Since also $\vz_{m,i}(d_i) = 0$, we establish the inductive case for~\ref{ind:zd0}.

% Induction hypothesis~\ref{ind:pmvalid} implies that $p_{m,i-1}$ is a valid transition sequence for $\vx_m + \ve_i$, since it is valid for $\vx_m + \ve_{i-1}$ by the hypothesis and $\vx_m + \ve_i \geq \vx_m + \ve_{i-1}$.
% We must show that all of $p_{m,i}$ is valid for $\vx_m + \ve_i$ by showing the final $c_i \cdot \alpha_i$ transitions on $p_{m,i}$, which apply to configuration $\vz_{m,i-1}(d_i)$, are valid.
% By induction hypothesis~\ref{ind:zdc}, there is sufficient count of $d_i$ to apply the transitions.
% By the fact that $\ve_i - \ve_{i-1} = \{c_i s_i\}$, there is sufficient count of the other input state $s_i$ to apply the transitions.
% This establishes the inductive case for~\ref{ind:pmvalid}.

Inductive case~\ref{ind:zdc} is proven as follows.
Let $j \in \{i+1,\ldots,k\}$. 
Induction hypothesis~\ref{ind:zdc} gives that $\vz_{m,i-1}(d_j) = \vy_m(d_j) + \left( \sum_{\ell=1}^{i-1} \{c_\ell o_\ell,c_\ell o'_\ell\} \right)(d_j)$.
Let $\vb = \{c_i o_i,c_i o'_i\}$ be the new term in the sum for the inductive case $i$; 
we must show that $\vz_{m,i} = \vz_{m,i-1} + \vb$.
Then $\vb(d_j)=0$ if $d_j$ is not an output state of $\alpha_i$, $\vb(d_j)=c_i$ if $d_j$ is exactly one output state, and $\vb(d_j)=2 c_i$ if $d_j$ is both output states.
Thus after applying $c_i \cdot \alpha_i$ to $\vz_{m,i-1} + \ve_i$ to result in configuration $\vz_{m,i}$, we have increased the count of $d_j$ by exactly $\vb(d_j)$, resulting in $\vz_{m,i}(d_j) = \vy_m(d_j) + \left( \sum_{\ell=1}^{i} \{c_\ell o_\ell,c_\ell o'_\ell\} \right) (d_j)$, proving the inductive case for~\ref{ind:zdc}.

To complete the proof we let $\vz_m^\Gamma = \vz_{m,k}$, 
for the final value $k$.
Inductive case~\ref{ind:zd0} on this final value $k$, 
shows that for all $j \in \{1,\ldots,k\}, \vz_{m,k}(d_j) = \vz^\Gamma_m(d_j) = 0$, thus proving that $\vz^\Gamma_m \in \N^\Gamma$.
Finally, note that $\ve = \sum_{i=1}^k \ve_i$ in the statement of the claim.
\qed\end{proof}

\newcommand{\cs}{c_s}
\newcommand{\cb}{c_b}  %c_0 in speed faults paper

Intuitively, Claim~\ref{claim:secondfixing} below works toward generating the vector of states $\ve$ that we needed for Claim~\ref{claim:firstfixing}.
The ``cost'' for Claim~\ref{claim:secondfixing} is that the path must be taken ``in the context'' of additional agents in states captured by $\vp$.
Importantly, the net effect of the path preserves $\vp$,
which will give us a way to ``interleave'' Claims~\ref{claim:firstfixing} and \ref{claim:secondfixing} as shown in Claim~\ref{claim:interleaving}.

\begin{clm}  \label{claim:secondfixing}
For all $\ve \in \N^\Lambda$,
there is $\vp \in \N^\Lambda$,
such that for all large enough $m$,
there is $\vw^\Gamma_m \in \N^\Gamma$,
such that
$\vp + \vx_m   \reach    \vp +   \vw^\Gamma_m + $ $\ve$.
\end{clm}

\paragraph{Example.}
Recall the example above illustrating Claim~\ref{claim:firstfixing}.
Claim~\ref{claim:secondfixing} is more difficult than Claim~\ref{claim:firstfixing} for two reasons.
First, we need to be able to obtain any counts of states $a$, $b$, $c$, $f$ in $\ve$, 
and not only ensure that $a = b = c = 0$.
Second, we no longer have the freedom to consume extra states (i.e., $\ve$ in Claim~\ref{claim:firstfixing}).
Note that $\vp$ cannot fulfill the same role as $\ve$ did in Claim~\ref{claim:firstfixing} because $\vp$ must be recovered at the end.
%Also, since the desired counts in $\ve_\Delta$ may be \emph{larger} than in $\vy_m$ (unlike in Claim~\ref{claim:firstfixing} where the desired counts of all states in $\Delta$ is 0), we do not necessarily have the option to add a transition to reduce the count; we may need to \emph{remove} a transition to \emph{increase} the count of a state.

For concreteness, suppose $\ve$ consists of $a=7$, $b=2$, $c=0$, $f = 3$.
To start with, note that handling state $f$ in $\ve$ is easy: recall $f$ is in $\Gamma = \Lambda\setminus\Delta$ and is present in ``large'' count in $\vy_m$. 
We can simply ``siphon'' the required number of agents in state $f$ into $\ve$ leaving the rest as $\vw_m^\Gamma$. 
For the rest of $\ve$, recall that $\vy_m$ has  $a=3$, $b=2$, $c=1$.
How can we generate an additional 4 copies of $a$?
Note that all transitions preserve or decrease the sum $a+b+c$.
Thus we cannot solely add interactions to $p_m$ to get to our desired $\ve$.
The key is that we can increase $a$ by removing existing interactions from $p_m$ that consumed it.
Indeed, Lemma~\ref{lem:ordering} helps us by giving a lower bound on the number of instances of transitions \eqref{rxn:claims-example-1},\eqref{rxn:claims-example-4},\eqref{rxn:claims-example-5} that must have occurred in $p_m$.
(Note that in Claim~\ref{claim:firstfixing}, we didn't need to use the fact that these transitions occurred in $p_m$.
Now, we need to ensure that there are enough instances for us to remove.)

In our case, to increase $a$ by 4, we can remove 4 instances of interaction \eqref{rxn:claims-example-1} from $p_m$,
resulting in $a=7, b=6, c=-3$.\footnote{Note the need for $\vp$ to ensure that the total count never goes negative. In writing ``$a=7, b=6, c=-3$'', we are examining the effect only on $\vy_m$ of modifying $p_m$, but in applying the lemma, the starting configuration is $\vx_m+\vp$, not merely $\vx_m$, so the actual count of each state $s$ will be $\vp(s)$ larger than just stated.}
To get $c=0$ as desired, we can remove 3 instances of transition~\eqref{rxn:claims-example-4}, 
resulting in $a=7,b=3,c=0$.
Finally, we add $1$ instance of transition~\eqref{rxn:claims-example-5} to get $a=7,b=2,c=0$ as desired.
%The net result is that we reach the configuration $a=7$, $b=2$, $c=1$, $f=130$.

Note that unlike in Claim~\ref{claim:firstfixing}, we have more potential for circularity now because we cannot add the other input to a transition as $\ve$.
For example, we can't use transition \eqref{rxn:claims-example-3} to affect $c$ because it affects $a$ (which we have previously driven to the desired count).
Luckily, the ordering given by Lemma~\ref{lem:ordering} avoids any circularity because the other input and both of the outputs come later in the ordering.

Importantly, as we remove or add interactions to $p_m$, we could potentially drive the count of some state negative---but only temporarily because the \emph{final} counts ($\vw^\Gamma + \ve$) are nonnegative.
Performing these interactions in the context of more agents ($\vp$) ensures that the path is valid.

%This tells us that we'll need another transition to drive one of them to 0.
%Suppose we add~\eqref{rxn:claims-example-1} to drive $a$ to 0.
%If we then try to use~\eqref{rxn:claims-example-3} to drive $c$ to 0, we fail for a similar reason as before, but subtly different: we already consumed all $d$, so there are none left to execute~\eqref{rxn:claims-example-3}.
%
%If we use a transition consuming a state $s$ to drive it to 0, $s$ has a dependency on the other three states appearing in the transition, in the sense that we cannot already have driven them to 0.
%In both cases above, what went wrong is that we encountered a cycle in this dependency graph.
%To drive each to 0, we must find some ordering on them so that each can be driven to 0 using a transition that does not depend on any state previously driven to 0.
%This is precisely what Lemma~\ref{lem:ordering} gives us.
%In the example above, we can choose $d_1=d, d_2=d^*$, and $d_3=\hd$, with respective transitions~\eqref{rxn:claims-example-1} to drive $d$ to 0 (by adding 3 extra executions of it), \eqref{rxn:claims-example-4} to drive $d^*$ to 0 (by adding 6 extra executions of it, 3 to consume the 3 copies of $d^*$ in $\vy$, and 3 more to consume the extra 3 copies that were produced by the 3 extra executions of~\eqref{rxn:claims-example-1}), and \eqref{rxn:claims-example-5} to drive $\hd$ to 0 (3 extra executions).

%\vspace{10pt}

\begin{proof}[of Claim~\ref{claim:secondfixing}]
  Define $\ve^\Delta \in \N^\Delta$ by $\ve^\Delta(d) = \ve(d)$ for all $d \in \Delta$ and $\ve^\Delta(s) = 0$ for all $s \in \Gamma$, and define $\ve^\Gamma \in \N^\Gamma$ similarly, so that $\ve = \ve^\Delta + \ve^\Gamma$.
  (Recall that $\Delta$ and $\Gamma$ partition the set of states $\Lambda$.)
  We proceed by proving the following claim, which focuses only on the $\ve^\Delta$ part of $\ve$, and which additionally ensures that $\vw_m^\Gamma$ grows on all states in $\Gamma$ as $m \to \infty$:

 \noindent \textbf{(*)} \emph{For all $\ve^\Delta \in \N^\Delta$, there is $\vp \in \N^\Lambda$, such that for all large enough $m$, there is $\vw^\Gamma_m \in \N^\Gamma$, such that $\vp + \vx_m \reach \vp + \vw_m^\Gamma + \ve^\Delta$ and $\unbdd((\vw_m^\Gamma)) = \Gamma$.
 }

  Supposing this claim is true, it is easy to complete the proof of Claim~\ref{claim:secondfixing} by handling positive $\ve^\Gamma$ as follows.
  Since $\unbdd((\vw_m^\Gamma)) = \Gamma$, 
  given any $\ve^\Gamma$, $\vw_m^\Gamma-\ve^\Gamma$ is non-negative for large enough $m$.
%  Indeed, for any $\ve^\Gamma$, for large enough $m$, $\vw_m^\Gamma-\ve^\Gamma$ is positive, which means that we can ``siphon'' the remaining states $\ve^\Gamma$ from $\vw_m^\Gamma$ to produce $\ve = \ve^\Delta + \ve^\Gamma$ as desired.}
%   \repds{Supposing that this can be shown, to see that it holds for \emph{any} $\ve \in \N^\Lambda$,
%   apply the claim to $\ve^\Delta$ to obtain that 
%   $\vp + \vx_m \reach \vp + \vw^\Gamma_m + \ve^\Delta$ 
%   and $\unbdd((\vw_m^\Gamma)) = \Gamma$.
%   Since $\|\ve^\Gamma\|$ is constant with respect to $m$,
%   we have that 
%   $\unbdd((\vw_m^\Gamma-\ve^\Gamma)) = \Gamma$ as well,
%   and for sufficiently large $m$,}{
  Then $\vp + \vx_m \reach \vp + \vw^\Gamma_m + \ve^\Delta = \vp + (\vw^\Gamma_m - \ve^\Gamma) + \ve$.
  In other words we apply claim (*) to produce $\ve^\Delta$, and then ``siphon'' the remaining states $\ve^\Gamma$ from $\vw_m^\Gamma$ to produce $\ve = \ve^\Delta + \ve^\Gamma$, which maintains the required conclusions of Claim~\ref{claim:secondfixing} with $(\vw_m^\Gamma-\ve^\Gamma)$ replacing $\vw_m^\Gamma$.
  
  We now show how to prove the above claim (*).
  Recall that $\vx_m \reach_{p_m}  \vy_m$ and $p_m$ does not contain an $m$-bottleneck.
  Intuitively, we will try to modify $p_m$ so that in the end we get exactly $\ve^\Delta$ of $\Delta$.
  %(Recall that $\Delta$ and $\Gamma$ partition the set of states $\Lambda$.)
  As in the proof of Claim~\ref{claim:firstfixing}, we will use the fact that $p_m$ does not contain an $m$-bottleneck and Lemma~\ref{lem:ordering} to find transitions affecting $\Delta$ in a non-circular manner.
  However, unlike in  Claim~\ref{claim:firstfixing}, we cannot simply consume 
  additional states (i.e., $\ve \in \N^\Lambda$ in Claim~\ref{claim:firstfixing}) to ensure that the count of the ``other input state $s_i$'' does not become negative.
  Rather, to increase the amounts of $s_i$ we will remove certain transition instances originally in $p_m$.
  It turns out that even with removing transitions, our modification to $p_m$ may still \emph{temporarily} take certain states negative if we start from $\vx_m$.
  However, executing the path in the context of $\vp$ provides ``buffer room'' to ensure that no counts ever go below zero.

  More formally, as in the proof of Claim~\ref{claim:firstfixing} apply Lemma~\ref{lem:ordering} with $b_2 = m$ and let $b_1$ be the largest count of any state in $\Delta$ anywhere in the sequence $(\vy_m)$.
  The lower bound on $b_2 = m$ is determined below (``bound on the amount of fixing'').
  Lemma~\ref{lem:ordering} tells us that there is an ordering on $\Delta$, so that we can write $\Delta=\{d_1,\ldots,d_k\}$, such that for each $1 \leq i \leq k$, there is a transition $\alpha_i: d_i,s_i \to o_i,o'_i$ such that $d_i \in \Delta$ and $s_i,o_i,o'_i \not\in \{d_1,\ldots,d_i\}$, and $\alpha_i$ occurs at least $(b_2-|\Lambda|\cdot b_1)/|\Lambda|^2$ times in $p_m$.
  Note that the final condition was not necessary to prove Claim~\ref{claim:firstfixing} since its proof only added transitions to $p_m$.
  However, since the current proof removes transitions as well, we require this condition to ensure that there are sufficiently many existing instances to be removed.

  %We now describe how to modify the transition sequence $p_m$ such that the amount of states in $\Delta$ after applying the modified path is exactly $\ve^\Delta$,
  %resulting in the desired path for all large enough $\vp \in \N^\Lambda$: $\vp + \vx_m   \reach    \vp +   \vw^\Gamma_m + \ve^\Delta$.

  We iteratively fix the counts of states in $\Delta$ one by one, in the ordering given,
  i.e. we first adjust $p_m$ to fix $d_1$, then we fix $d_2$ (while showing that the fixing of $d_2$ cannot affect the count of $d_1$ in any configuration, so it remains fixed), etc.
  We start with $\ve^\Delta_0(s) = \vy_m(s)$ for $s \in \Delta$ and $\ve^\Delta_0(s)=0$ for $s \in \Gamma$, and $\vw^\Gamma_{m,0}(s) = \vy_m(s)$ for $s \in \Gamma$ and $\vw^\Gamma_{m,0}(s) = 0$ for $s \in \Delta$.
  Having fixed $d_1, \dots, d_{i-1}$, and obtaining new $\ve^\Delta_{i-1}$, $\vw^\Gamma_{m,i-1}$ (which could now be negative) such that $\ve^\Delta_{i-1}$ agrees with the desired $\ve^\Delta$ over $d_1, \dots, d_{i-1}$,
  we process $d_i$ as follows.
  If $\delta_i = \ve^\Delta(d_i) - \ve^\Delta_{i-1}(d_i) < 0$: add $\delta_i$ instances of transition $\alpha_i$ at the end of the transition sequence.
  If $\delta_i > 0$: remove $\delta_i$ instances of $\alpha_i$ where they occur in the transition sequence;
  property (3) ensures that $q$ contains enough instances of $\alpha_i$ (see below).
  Let $\ve^\Delta_{i}$ be the counts of the states in $\Delta$ at the end of this path.
  By property (2) and (3), adding or removing instances of $\alpha_i$ affects only the counts of states in $\Gamma$ and $d_{i+1},\ldots,d_k$.
  Since we fix these counts in the prescribed order, when we are done, the counts of each $d_i$ is equal to its count in $\ve^\Delta$ (ie $\ve^\Delta = \ve^\Delta_k$),
  while counts of elements of $\Gamma$ have been altered (letting $\vw_m^\Gamma = \vw_{m,k}^\Gamma$).
  We now claim that for large enough $m$, $\vw_{m,k}^\Gamma$ is nonnegative, and that $\vp$ can be independent of $m$.
  Finally, we derive a bound on the number of transition instances that we may need to remove, which determines another bound on $m$ (ie $b_2$) to ensure that there are enough instances by property (4) above.

  Note that the amount of fixing we need to do only depends on the desired $\ve^\Delta$ as well as on the counts of $\Delta$ states in $\vy_m$.
  Because $\vy_m$ are nondecreasing, and $\Delta = \bdd((\vy_m))$, for large enough $m$, the counts of $\Delta$ states in $\vy_m$ stop changing,
  and the amount of fixing depends only on the desired $\ve^\Delta$.
  This implies that the $\vp$ we need to add to ensure that no counts go negative can be independent of $m$.
  Further, for large enough $m$, the difference between $\vw^\Gamma_m$ and $\vy_m$ is independent of $m$,
  and thus $\unbdd(\vw^\Gamma_m) = \unbdd(\vy_m) = \Gamma$, as needed for claim~(*).
  This also implies that for large enough $m$, $\vw_{m,k}^\Gamma$ is nonnegative.

  \textit{Bound on the amount of fixing:}
  %\repds{}{Although proving the claim does not require an explicit calculation on the number of transition instances added or removed, only that it can be bounded independently of $m$}
  We now derive a bound the number of transition instances that must be added/removed, in order to justify that this bound depends only on $\ve$, but is independent of $m$.
  Define the quantity $\cb = \max\limits_{m\in\N,d\in\Delta} |\vy_m(d) - \ve^\Delta(d)|$ as the maximum amount that  any state in $\Delta$ deviates from its desired count. %, and let $\cs$ is the maximum stoichiometric coefficient.
  We add or remove at most $|\delta_1| \leq \cb$ instances of $\alpha_1$, which affects the count of states in $\Gamma \cup \{d_2,\ldots,d_k\}$ by at most $2 \cb$ (it could be 2 per transition if the transition is $\alpha_1: d_1,s \to s',s'$ for some state $s' \in\Lambda$).
  Thus, $|\delta_2| \leq \cb + 2 |\delta_1| $ (the original $\cb$ error plus the additional error from altering the number of $\alpha_1$ transitions).
  In general, $|\delta_i| \leq \cb + 2 (|\delta_1| + \cdots + |\delta_{i-1}|)  \leq  3^{i-1} \cb$.
  Thus if we let $m = b_2 \geq k \cdot b_1 + 3^{k-1} \cb |\Lambda|^2$, we will have enough transition instances by property (4) to remove ($(b_2-|\Lambda|\cdot b_1)/|\Lambda|^2 = 3^{i-1} \cb$).
\qed\end{proof}

\shrinkBeforePar
\begin{clm}  \label{claim:interleaving}
For infinitely many $\vi \in I$, there is $\vv^\Gamma \in \N^\Gamma$ such that $\vi \reach \vv^\Gamma$.
\end{clm}

\begin{proof}
Intuitively, Claim~\ref{claim:interleaving} follows by expressing $\vi = 2\vi'$ where $\vi' \in I'$ and $\vi' \reach \vx_{m}$,
so $2 \vi' \reach 2 \vx_{m}$.
We then apply Claim~\ref{claim:secondfixing} to one copy of $\vx_{m}$ (with the other $\vx_{m}$ playing the role of $\vp$) to get to a configuration with the correct $\ve$ for Claim~\ref{claim:firstfixing},  and then apply Claim~\ref{claim:firstfixing} to remove all states in $\Delta$.

%\todoi{DD: I'm not sure if we should make more explicit why $m_1,m_2$ can be chosen arbitrarily large; it's because $\|\vi_1\|$ and $\|\vi_2\|$ can, and by Theorem~\ref{thm:timer} $m_i$ grows with $\|\vi_i\|$, but I'm not sure if that's distracting to get into at this point. Maybe it should be stated as a consequence of Lemma~\ref{lem:pos-prob-states-expected-time}?}

  Choose $m$ large enough to satisfy the conditions stated below as they are needed. 
  By Claim~\ref{claim:firstfixing}, there is $\ve \in \N^\Lambda$ and $\vz^\Gamma \in \N^\Gamma$ such that $\vx_{m} + \ve \reach_{p_{m}} \vz^\Gamma$.
  Apply Claim~\ref{claim:secondfixing} on $\ve$ (making sure $m$ is large enough to satisfy the claim on $\ve$).
  Thus, there is $\vp \in \N^\Lambda$ and $\vw^\Gamma \in \N^\Gamma$ such that
  $\vp + \vx_{m}   \reach_{p_{m}} \vp + \vw^\Gamma + \ve$.
  If $m$ is large enough that $\vx_{m} \geq \vp$, then
  $\vx_{m} + \vx_{m}  \reach_{p_{m}}  \vx_{m} + \vw^\Gamma + \ve$.
  Then, by Claim~\ref{claim:firstfixing},
  $\vx_{m} + \vw^\Gamma + \ve \reach_{p_{m}} \vw^\Gamma + \vz^\Gamma$.
  To complete the claim, we let $\vv^\Gamma = \vw^\Gamma + \vz^\Gamma.$
\qed\end{proof}

%\vspace{8pt}
Finally, Theorem~\ref{thm:main-q-locking} is proven because $\vv^\Gamma$ is $Q$-stable and it  contains zero count of states in $\Delta$.
To see that $\vv^\Gamma$ is $Q$-stable recall that $\vv^\Gamma \leq \vy_{m'}$ for sufficiently large $m'$ since $\Gamma = \unbdd{((\vy_m))}$ and $\vv^\Gamma$ contains only states in $\Gamma$.
Since stability is closed downward, and $\vy_{m'}$ is $Q$-stable, we have that $\vv^\Gamma$ is $Q$-stable as well.

%\opt{normal,sub,conf}{\ack}

\opt{sub}{\newpage}
\opt{sub,conf}{\bibliographystyle{splncs03}}
\opt{normal}{\bibliographystyle{plain}}
\bibliography{tam}

\opt{sub}{
  %\newpage
  \appendix
  \section{Appendix}

  \subsection{\difficultyTitle} \label{sec:difficulty}
  \difficulty

  %\subsection{Proofs omitted from main text}
  \subsection{Proofs Omitted from Main Text}

  \reobs{\ref{obs:bottleneck-linear-time}}{\obsBottleneck}
  \proofObsBottleneck

  \relem{\ref{lem:pos-prob-states-expected-time}}{\lemPosProbExpTime}
  \begin{proof}
    \proofLemPosProbExpTime
  \qed\end{proof}

  \relem{\ref{lem:ordering}}{\transitionOrdering}
  \proofTransitionOrdering

  Recall the sequences $(\vx_m)$, $(\vy_m)$, and $(p_m)$, and the sets of states $\Delta = \bdd((\vy_m))$ and $\Gamma = \unbdd(\vy_m)$, defined just before Claim~\ref{claim:firstfixing} in Section~\ref{sec:main-proof}.

  \reclaim{\ref{claim:firstfixing}}{\claimFirstfix}
  \begin{proof}
  \proofClaimFirstfix
  \qed\end{proof}

  \reclaim{\ref{claim:secondfixing}}{\claimSecondfix}
  \begin{proof}
  \proofClaimSecondfix
  \qed\end{proof}

  \reclaim{\ref{claim:interleaving}}{\claimInterleaving}
  \begin{proof}
  \proofClaimInterleaving
  \qed\end{proof}

}
\end{document}